\documentclass[journal,comsoc]{IEEEtran}

\usepackage[T1]{fontenc}
\ifCLASSINFOpdf
\else
\fi
\usepackage[pdftex]{graphicx}
\usepackage{amsmath,amssymb,amsthm}
\usepackage{subfig}
\usepackage{bm}
\usepackage{algorithm}
\usepackage{algpseudocode}
\usepackage{changepage}

\usepackage[cmintegrals]{newtxmath}

\newtheorem{theorem}{Theorem}
\newtheorem{lemma}{Lemma}
\newtheorem{proposition}{Proposition}

\newtheorem{rmk}{Remark}

\theoremstyle{definition}

\newtheorem{example}{Example}

\begin{document}
%
\title{Capacity of Distributed Storage Systems with Clusters and Separate Nodes}


\author{Jingzhao~Wang,
        Tinghan~Wang,
        Yuan~Luo,~\IEEEmembership{Member,~IEEE,}
        and~Kenneth W. Shum~\IEEEmembership{Senior Member,~IEEE}
\thanks{J. Wang, T. Wang and Y. Luo are with the Department of Computer Science and Engineering, Shanghai Jiao Tong University, Shanghai, 200240, China. Email: \{wangzhe.90, wth19941018, yuanluo\}@sjtu.edu.cn}
\thanks{Kenneth W. Shum is with the Institute of Network Coding, The Chinese University of Hong Kong, Shatin, New Territories, Hong Kong. Email: wkshum@inc.cuhk.edu.hk}
}


%


\maketitle

\begin{abstract}
 In distributed storage systems (DSSs), the optimal tradeoff between node storage and repair bandwidth is an important issue for designing distributed coding strategies to ensure large scale data reliability. The capacity of DSSs is obtained as a function of node storage and repair bandwidth parameters, characterizing the tradeoff. There are lots of works on DSSs with clusters (racks) where the repair bandwidths from intra-cluster and cross-cluster are differentiated. However, separate nodes are also prevalent in the realistic DSSs, but the works on DSSs with clusters and separate nodes (CSN-DSSs) are insufficient. In this paper, we formulate the capacity of CSN-DSSs with one separate node for the first time where the bandwidth to repair a separate node is of cross-cluster. Consequently, the optimal tradeoff between node storage and repair bandwidth are derived and compared with cluster DSSs. A regenerating code instance is constructed based on the tradeoff. Furthermore, the influence of adding a separate node is analyzed and formulated theoretically. We prove that when each cluster contains $R$ nodes and any $k$ nodes suffice to recover the original file (MDS property), adding an extra separate node will keep the capacity if $R|k$, and reduce the capacity otherwise.
\end{abstract}


%
\IEEEpeerreviewmaketitle

\section{Introduction}
In the age of big data, massive amount of data are generated and stored in large data centers every day, where ensuring the data reliability is an important issue\cite{Hu2017}. Erasure coding is widely used to tolerate node failures in distributed storage systems (DSSs)\cite{Aggarwal2017Sprout,Abbasi2018Video,Huang2013,Li2017Beehive,Zhang2017Tree}, where the original data are encoded and stored in multiple nodes. If a node failure happens, a newcomer is generated by downloading data from other nodes, which may incur high repair bandwidth\cite{Xiang2016Joint}. The authors of \cite{Dimakis2010} proposed regenerating codes to balance the node storage and repair bandwidth, which is characterized by the capacity of DSSs with homogeneous node parameters \cite{Ernvall2013}, meaning that all the storage nodes are undifferentiated.

On the other hand, in heterogeneous DSSs~\cite{Ernvall2013,yu2015tradeoff}, the storage and repair bandwidth parameters are different for different nodes based on the variety of real storage devices. In realistic storage systems, nodes are generally grouped with clusters (racks)\cite{ford2010}, where the intra-cluster networks are often faster and cheaper. In order to make use of the communication resource efficiently, it is necessary to differentiate intra-cluster and cross-cluster bandwidths. Additionally, the cross-cluster bandwidth is often constrained in modern data centers \cite{Hu2017}. For instance the available cross-cluster bandwidth for each node is only $1/5$ to $1/20$ of the intra-cluster bandwidth in some cases \cite{ahmad2014shuffle,benson2010network,vahdat2010Scale}. It is efficient and practical to consider the balancing problem of the storage and repair bandwidth under realistic network topology. In \cite{tree2010Li,heteroAware2014Wang}, the authors proposed tree-based topology-aware repair schemes considering networks with heterogeneous link capacities. In \cite{Sipos2018NetworkAware}, Sipos \emph{et al.} proposed a general network aware framework to reduce the repair bandwidth in heterogeneous and dynamic networks.

In \cite{Hou2018Rack,Hu2016double, Hu2017}, the authors investigated multi-cluster (rack) models to reduce the cross-cluster bandwidth, where data from nodes in each cluster were collected and transmitted with a relay node to repair a fail one.  Coding strategies were proposed to minimize the cross-cluster repair bandwidth, which were deployed to verify the performance in hierarchical data centers in \cite{Hou2018Rack,Hu2017}. In \cite{Abd2017storage,Prakash2017}, the authors also investigated cluster DSSs with relay nodes which were not only used for node repair, but also used for data collection. In \cite{shao2018RackAware}, the authors proposed a data placement method to reduce the cross-cluster bandwidth on data reconstruction in the cluster DSS model without relay nodes. The DSS model with two clusters was considered in \cite{Akhlaghi2010,Two_rack2013}. In \cite{sohn2017capacity}, the authors first proposed algorithms to characterize the capacity of the cluster DSS model also with no relay nodes under the assumption that all the other alive nodes are used to repair a failed one, which maximizes the system capacity as proven in \cite{sohn2018capacity}. However, this led to high \emph{reconstruction read cost} which was defined in \cite{Huang2013} as the number of helper nodes to recover a failed one. It would be practical and flexible to consider the cluster DSS model where the number of helper nodes is not restricted to the maximum. Although the system capacity would be smaller, we gain more flexibility in the repair process. For example, multiple node failures can be analysed under flexible repair constrains, which was also considered in \cite{Abd2017storage}.

In our earlier paper \cite{Wangjz2018}, we analysed the capacity of the cluster DSS model under more flexible constraints, where the newcomer did not have to download data from all the other alive nodes. In addition, the storage servers (nodes) and networks vary in realistic storage systems, which is the motivation for researching heterogeneous DSSs \cite{yu2015tradeoff,Sipos2018NetworkAware}. As a general model for cluster DSSs, it is practical to consider the CSN-DSS model \cite{Kub2000Ocean,yu2015tradeoff}, which is instructive for construct efficient coding scheme adaptive to various networks. In \cite{Wangjz2018}, we introduced and analysed partly the CSN-DSS model. However, the properties of CSN-DSSs are not investigated in detail. The final system capacity and tradeoff with separate nodes are not characterized either, which will be formulated in the present paper.

In Section~\ref{sec_Pre}, the CSN-DSS model is introduced, where the capacity and tradeoff problems are formulated. In Section~\ref{sec:NoSN}, we sketch the properties of cluster DSSs proved in \cite{Wangjz2018}, which are also useful in analysing CSN-DSSs. The main contributions of this paper are as follows. The CSN-DSS model is analysed in Section~\ref{sec:withS}, where we prove the applicability of Algorithm~1 and 2 when adding a separate node in Theorem~\ref{theorem_MC3}. When the location of the added separate node varies, the capacity is analysed in Theorem~\ref{thm_MCj}. Consequently, the final capacity of the CSN-DSS model is derived in Theorem~\ref{theorem_MCSN}. Afterward, the tradeoffs between node storage and repair bandwidth are characterized for the cluster DSS and CSN-DSS models in Section~\ref{sec_tradoff}. Based on the tradeoff bounds, a regenerating code construction for the CSN-DSS model is investigated in Section~\ref{sec:consWithS}. In Section~\ref{sec_compare}, we analyse the influence of adding a separate node to the system capacity theoretically. We prove that adding a separate node will reduce or keep the capacity of a cluster DSS, depending on the system node parameters.

\section{Preliminaries}\label{sec_Pre}

\subsection{The model of distributed storage system with clusters and separate nodes (CSN-DSS)}\label{subsec_CSND}

\begin{figure}
  \centering
  \includegraphics[width=0.25\textwidth]{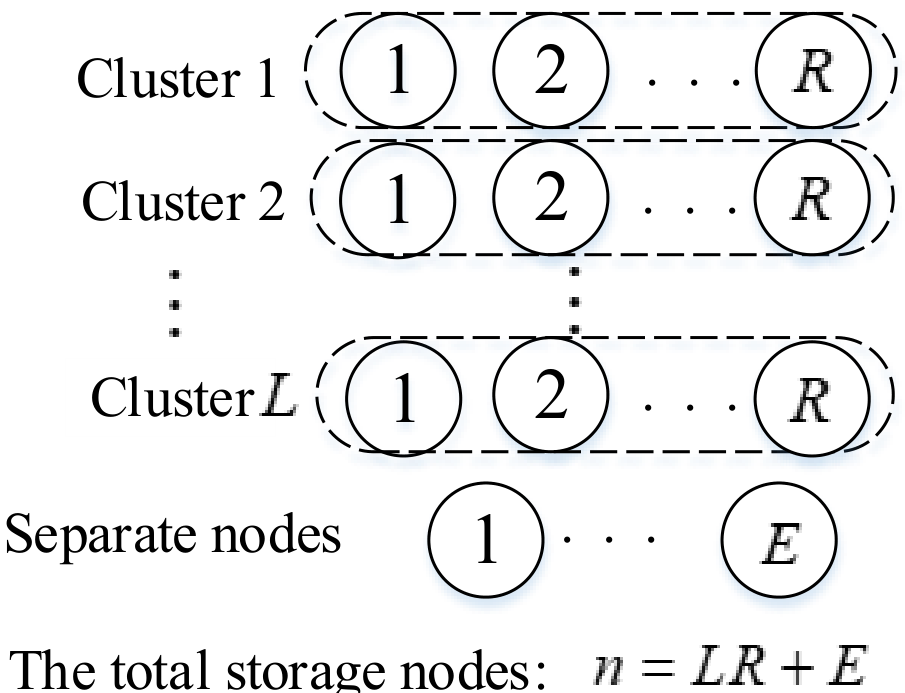}
  \caption{The CSN-DSS system model}\label{fig_CSNmodel}
\end{figure}

As Figure~\ref{fig_CSNmodel} shows, the CSN-DSS model consists of $n=LR+E$ storage nodes in total ($L$ clusters and $E$ separate nodes). Each cluster contains $R$ nodes. Assume the original data of size $\mathcal{M}$ are encoded with erasure coding and stored in $n$ nodes each of size $\alpha$. The $n$ nodes satisfy the $(n,k)$ MDS\footnote{Maximum distance separate (MDS) codes achieve optimality in terms of redundancy
and error tolerance} property meaning that any $k$ nodes out of $n$ suffice to reconstruct the original data. When a node fails, there are two types of repair pattern: \emph{exact repair} and \emph{functional repair}. In exact repair, the lost data must be exactly recovered. On the other hand, we only demand the $n$ nodes after each repair keep the MDS property \cite{Survey2011} in functional repair. We assume the repair procedure will not change the location of failed nodes. Thus the number of model nodes will not change with node failure and repair. This paper handles the functional repair situation and only considers one node failure.

When repairing a \textbf{cluster node}, the newcomer downloads $\beta_I$ symbols from each of the $d_I$ intra-cluster nodes and $\beta_C$ symbols from each of the $d_C$ cross-cluster nodes. We define the total number of helper nodes as
\begin{equation*}
  d\triangleq d_I+d_C.
\end{equation*}
The same as the restrictions introduced in\cite{Rashmi2011Optimal} inequality (1),  we also assume $k\leq d\leq n-1$ in the present paper.

As explained in \cite{Wangjz2018}, transmissions among intra-cluster nodes are much cheaper and faster, so it is natural to download more data from intra-cluster nodes and less from cross-cluster nodes, namely, $\beta_I\geq \beta_C$. Moreover, all the intra-cluster nodes are utilized in the repair procedure, namely $d_I=R-1$ in this paper, since using intra-cluster nodes is efficient and preferential in general cases.

As each separate node can be seen as a special cluster only with one node and all the other nodes are cross-cluster, we assume the newcomer downloads $\beta_C$ symbols from each of $d$ other nodes to repair a failed \textbf{separate node}. Note that only $d$ helper nodes are used no matter where the failed one is.

In the CSN-DSS model, $(\alpha, d_I, \beta_I, d_C, \beta_C)$ and $(n,k,L,R,E)$ are called the \textbf{storage/repair} and \textbf{node parameters} respectively for simplicity. Additionally, when repairing a fail cluster node, the total intra-cluster and is defined as
$$\gamma_I\triangleq d_I\beta_I.$$\noindent
Meanwhile, the cross-cluster bandwidth is defined as
$$\gamma_C \triangleq d_C\beta_C.$$
The total repair bandwidth for a separate node is $\gamma_S\triangleq d\beta_C.$ The traditional homogeneous DSS of \cite{Dimakis2010} is retrieved here if $\beta_I=\beta_C$.

\subsection{Information Flow Graph (IFG)}\label{subsec_IFG}

\begin{figure}
  \centering
  \includegraphics[width=0.4\textwidth]{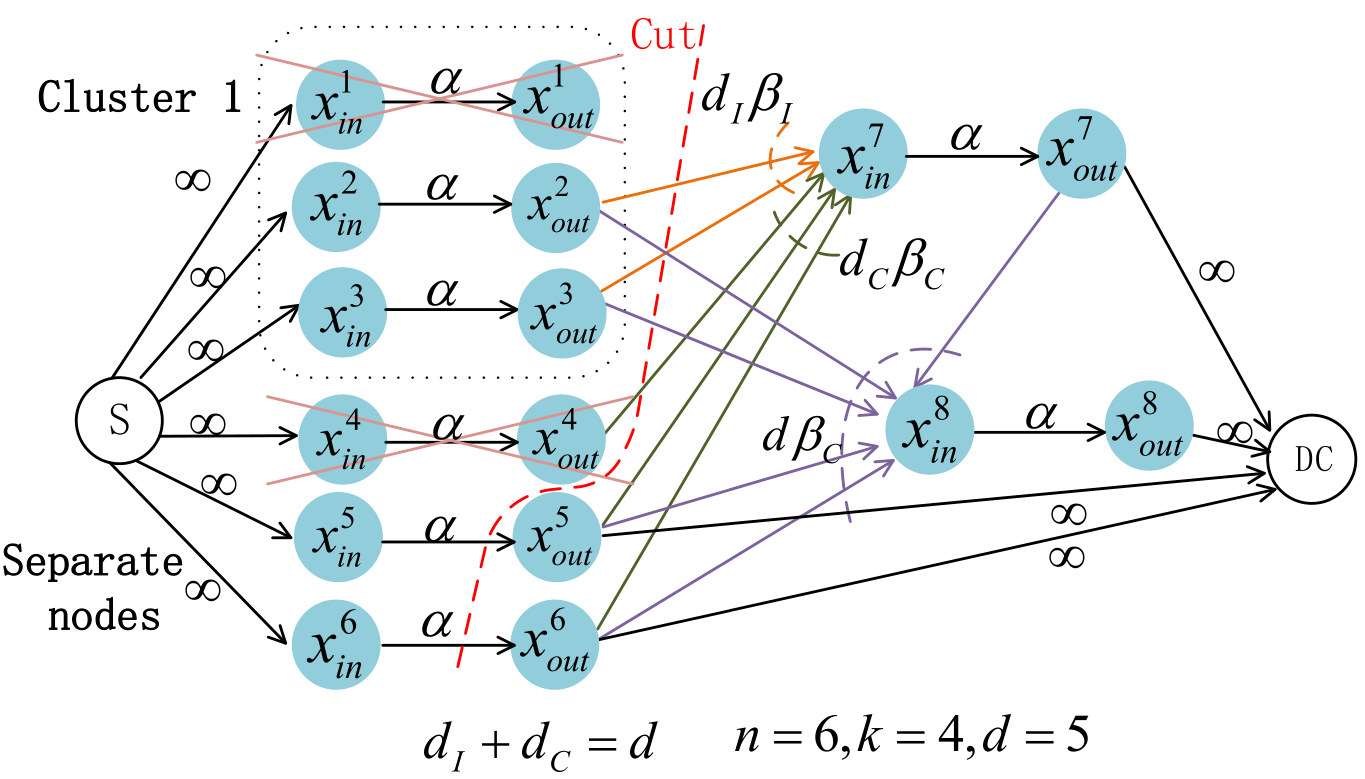}
  \caption{An IFG of the CSN-DSS model}\label{fig_IFGCSN}
\end{figure}

In~\cite{Dimakis2010}, the performance of DSSs was analysed with the information flow graph (IFG) consisting of three types of nodes: data source $S$, storage nodes $x_{in}^i$, $x_{out}^i$, and data collector $DC$ (see Figure~\ref{fig_IFGCSN}). We use $x^i$ to denote a physical storage node represented by a storage input node $x_{in}^i$ and an output node $x_{out}^i$ in the IFG, where data are pre-handled before transmission. The capacity of edge $x_{in}^i\rightarrow x_{out}^i$ is $\alpha$ (the node storage size).

At the initial time, source $S$ emits $n$ edges with infinite capacity to $\{x_{in}^i\}_{i=1}^n$, representing the encoded data are stored in $n$ nodes. Subsequently, $S$ changes to inactive, while the $n$ storage nodes become active. The node $x^j$ $(1\leq i \leq n)$ gets inactive if it has failed. In the failure/repair process, a new node $x^{n+1}$ is added by connecting edges with $d$ active nodes, where the capacity of each edge is $\beta$. The value of $\beta$ varies in the CSN-DSS model (see Figure~\ref{fig_IFGCSN}). The IFG maintains $n$ active nodes after each repair procedure. To keep the $(n,k)$ MDS property, $DC$ selects arbitrary $k$ active nodes to reconstruct the original date, as shown by the edges from $k$ active nodes with infinite capacity.

Figure~\ref{fig_IFGCSN} shows an IFG of the CSN-DSS model, where the cluster node $x^1$ has failed firstly. Then the newcomer $x^7$ is created, downloading $\beta_I$ symbols from each of $x^2$ and $x^3$ in cluster 1, and $\beta_C$ symbols from each of $x^4$, $x^5$ and $x^6$ out of cluster 1. Subsequently, $x^4$ is failed, $x^8$ is generated, connecting with five active nodes.

For an IFG, a (directed) cut between source $S$ and $DC$ is defined as a subset of edges, satisfying the condition that every directed path from $S$ to $DC$ contains at least one edge in the subset. The \textbf{min-cut} is the cut between $S$ and $DC$ where the total sum of the edge capacities is smallest \cite{Dimakis2010}.

\subsection{Problem Formulation}\label{subsec_PM}

As introduced in \cite{Wangjz2018}, for a CSN-DSS model with $(n,k,L,R,E)$, the main problem is to characterize the feasible region of points $(\alpha, d_I, \beta_I, d_C, \beta_C)$ to store file of size $\mathcal{M}$ reliably. Similarly to \cite{Dimakis2010}, this problem is solved through analysing the min-cuts of all possible IFGs corresponding to this CSN-DSS model. According to the max flow bound \cite{Ahlswede2006,Li2003} in network coding, to ensure reliable storage, the file size $\mathcal{M}$ will not greater than the system \textbf{capacity}
\begin{small}\begin{eqnarray*}
\mathbb{C}\triangleq \text{min-cut of }G^*,
\end{eqnarray*}\end{small}\noindent
where $G^*$ is the IFG with the minimum min-cut. That is to say
\begin{small}\begin{equation}\label{equ_bound}
\mathbb{C}\geq\mathcal{M}
\end{equation}\end{small}\noindent
need to be satisfied, which was also proved in \cite{Dimakis2010}. When the node parameters are given, $\mathbb{C}$ is obtained as a function of $(\alpha, d_I, \beta_I, d_C, \beta_C)$. Thus the tradeoff between $\alpha$ and $(d_I, \beta_I, d_C, \beta_C)$ can be derived.

\subsection{Terminologies and Definitions} \label{subsec_TermDef}
In this subsection, we introduce some terms defined in \cite{Wangjz2018} to investigate the capacity of the CSN-DSS model through analyse the min-cuts of the corresponding IFGs.

\textbf{Topological ordering and the min-cut:} As introduced in \cite{Networkflows1993}, the topological ordering of vertices in a directed acyclic graph is an ordering that if there exists a path from $v_i$ to $v_j$ then $i<j$. We use $\{x_{out}^{t_i}\}_{i=1}^k$ to denote the $k$ topologically ordered nodes connecting to $DC$. As proven in~\cite{Dimakis2010}, the minimum min-cut is achieved when nodes $\{x^{t_i}\}_{i=1}^k$ are all newcomers and $x^{t_i}$ connect to all the $i-1$ nodes $\{x^{t_j}\}_{j=1}^{i-1}$. In the CSN-DSS model, to find the minimum min-cut, we also assume that $x^{t_i}$ connect to all the former nodes $\{x^{t_j}\}_{j=1}^{i-1}$, which was also proved in \cite{sohn2018capacity}. The \textbf{min-cut} can be obtained by cutting $\{x^{t_i}\}_{i=1}^k$ one by one in the topological ordering, which is analysed in Subsection~\ref{subsec:Mincutcal}. In fact, when cutting $x^{t_i}$, since the $k$ output nodes connect to $DC$ with edges of infinite capacity, we only need to compare the capacity of edge $x_{in}^{t_i}\rightarrow x_{out}^{t_i}$ and the total capacity of the edges emanating to $x_{in}^{t_i}$. Then we choose the minor one for cutting, which is called a \textbf{part-cut value}.

For instance, the data collector connects to $x_{out}^6, x_{out}^5, x_{out}^7, x_{out}^8$ in Figure~\ref{fig_IFGCSN}, which are topologically ordered. Assume that the red dashed line is the final cut line. In this case, when cutting $x^7$, we need to compare the capacity of edge $x_{in}^{7}\rightarrow x_{out}^{7}$ and the total capacity of the edges emanating to $x_{in}^7$, where the latter one is minor if $2\beta_I+\beta_C\leq \alpha$. Note that, not all the edges emanating to $x_{in}^{7}$ are cut and counted in the part-cut value,
which depends on the topological ordering analysed in Subsection~\ref{subsec:Mincutcal}. On the other hand, if $2\beta_I+\beta_C > \alpha$, the cut line will cross the edge $x_{in}^7\rightarrow x_{out}^7$ and the part-cut value will be $\alpha$. When all the $k$ nodes connecting to $DC$ are cut, the min-cut is calculated by summing the $k$ part-cut values (see formula (\ref{equ_MC})).

\textbf{Repair sequence and selected nodes}:  The topological ordering of $k$ output nodes $\{x_{out}^{t_i}\}_{i=1}^k$ corresponds to a \textbf{repair sequence} of original nodes which are called \textbf{selected nodes}. For example, the numbered nodes in Figure~\ref{fig_sequence} are 7 selected nodes where the numbers indicate a repair sequence. In homogeneous distributed storage systems \cite{Dimakis2010}, the storage nodes are undifferentiated, thus the min-cuts are independent of repair sequences. However, due to the heterogeneity of intra-cluster and cross-cluster bandwidths in the CSN-DSS model, various repair sequences lead to different min-cuts, and we represent the repair sequence with following definitions.

\textbf{Selected node distribution and cluster order:} For a CSN-DSS model with $(n,k,L,R,E)$, we relabel the clusters by the amount of selected nodes in a non-increasing order without loss of generality. For example, in Figure~\ref{fig_sequence}, cluster~1 contains 3 selected nodes (the most) and cluster~3 contains 0 selected nodes (the least). The \textbf{selected node distribution} is denoted with $\textbf{s}=(s_0,s_1,...,s_L)$, where $s_0$ is the amount of separate selected nodes, and $s_i\ (1\leq i\leq L)$ represents the amount of selected nodes in cluster $i$. Meanwhile, the set of all possible selected node distributions is denoted as

\begin{align*}\label{equ_S}
  \mathcal{S}\triangleq\Big\{\textbf{s}=(s_0,s_1,...,s_L):&\ s_{i+1}\leq s_{i},\ 0\leq s_i\leq R,\ \\
  &\text{for}\ 1\leq i\leq L;\ 0\leq s_0\leq S;\ \sum_{i=0}^L s_i=k\Big\}.
\end{align*}
Additionally, the {\bf repair sequence} is represented by the \textbf{cluster order} $\bm{\pi}=(\pi_1,\pi_2,...,\pi_k)$, where $\pi_i\ (1\leq i\leq k)$ equals the index of the cluster containing newcomer $x^{t_i}$. We set $\pi_i=0$ when node $i$ is separate. As the nodes in one cluster are undifferentiated, we only need to record the cluster index. For a certain  $\textbf{s}=(s_0, s_1, s_2,..., s_L)$, the set of all possible cluster orders is specified as
\begin{align*}
\Pi(\textbf{s})\triangleq\Big\{\bm{\pi}=(\pi_1,..., \pi_k):\ \sum_{j=1}^k \mathbb{I}(\pi_j= i)=s_i\ \text{for }  0\leq i \leq L\Big\},
\end{align*}
\noindent
where
\begin{align*}\mathbb{I}(\pi_j=i)=\begin{cases}
                                   1, & \text{if\ } \pi_j=i,\\
                                   0, & \text{otherwise} .
                                  \end{cases}
\end{align*}

In Figure~\ref{fig_sequence}, $\textbf{s}=(1,3,3,0)$ represents that the $DC$ connects one separate node and  the nodes from cluster 1, 2 and 3 are 3, 3, 0 respectively. Meanwhile, a possible cluster order for $\textbf{s}$ is $\bm{\pi}=(1,1,1,2,2,2,0)$, indicated by the numbers from $1$ to $7$. As Figure~\ref{fig_sequence} shows, the $k$ selected nodes are numbered sequentially for convenience. Additionally, the cluster nodes are also ordered column by column cross clusters, for simplicity of description, we also say that node 1 and 4 are in the first column of the model, node 2 and 5 are in the second column, and node 3 and 6 are in the third column, meaning that the columns of cluster nodes are labeled by 1 to $R$ from left to right implicitly. These descriptions are used in the following sections.

Moreover, for a given cluster order $\bm{\pi}$, assume node $i$ is the $h_{\bm{\pi}}(i)$-th one in its own cluster, where we define the \textbf{relative location} of node $i$ as
\begin{equation}\label{equ_hi}
h_{\bm{\pi}}(i)\triangleq\sum_{j=1}^i \mathbb{I}(\pi_j=\pi_i),
\end{equation}\noindent
for $1\leq i\leq k$. Note that the \textbf{relative location} specifies the precedence of selected nodes in each cluster. It is obvious that node $i$ is in the $h_{\bm{\pi}}(i)$-th column. For instance, Figure~\ref{fig_sequence} illustrates $\bm{\pi}=(1,1,1,2,2,2,0)$ and the corresponding $h_{\bm{\pi}}(i)$ sequence $(1,2,3,1,2,3,1)$ calculated with (\ref{equ_hi}). For node $3$, $h_{\bm{\pi}}(3)=\mathbb{I}(\pi_1=\pi_3)+\mathbb{I}(\pi_2=\pi_3)+\mathbb{I}(\pi_3=\pi_3)=1+1+1=3$ where $\pi_3=1$, and therefore node $3$ is in the third column.

\begin{figure}[t]
    \centering
    \includegraphics[width=0.27\textwidth]{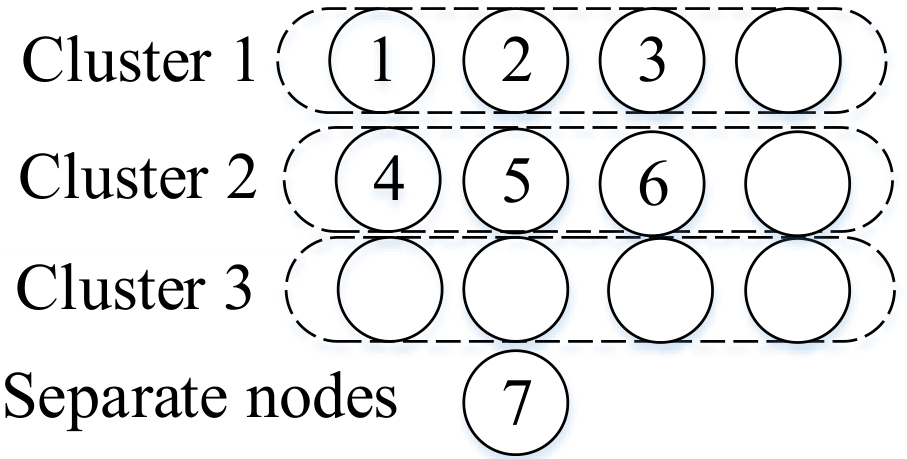}
    \caption{For a CSN-DSS model, the selected node distribution is $\textbf{s}=(1, 3, 3, 0)$ and the cluster order $\bm{\pi}=(1,1,1,2,2,2,0)$ as the numbered nodes show.} \label{fig_sequence}
  \end{figure}

\subsection{Calculating the min-cut value}\label{subsec:Mincutcal}

As explained in Subsection~\ref{subsec_TermDef}, for given $\textbf{s}$ and $\bm{\pi}$, the min-cut value is obtained by calculate the $k$ part-cut values step by step. When calculating the $i$-th part-cut value, we define the $i$-th part incoming weight with formula (\ref{equ_wi}). The final min-cut value is obtained by formula (\ref{equ_MC}).

\begin{figure}[t]
  \centering
    \includegraphics[width=0.3\textwidth]{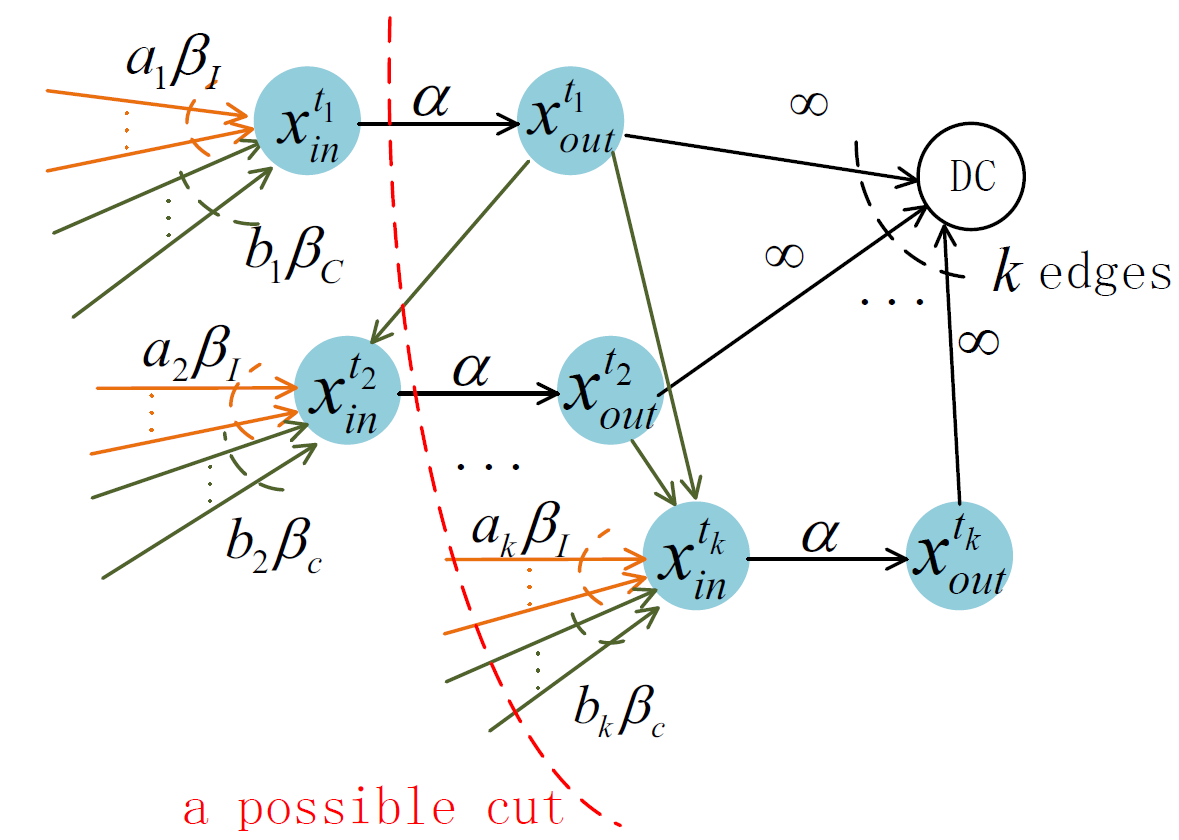}
    \caption{Calculating the min-cut value} \label{fig_partcut}
\end{figure}

As introduced in \cite{Dimakis2010,Wangjz2018}, after cutting the $k$ topologically ordered nodes $\{x^{t_i}\}_{i=1}^k$ (see Figure~\ref{fig_partcut}) iteratively, the nodes of the IFG $G$ are divided into two disjoint sets $V$ and $\overline{V}$. The min-cut (denoted by $\mathcal{U}$) is the set of edges emanating from $V$ to $\overline{V}$. In the beginning, $V$ consists of source $S$ and the original nodes $x^i(1\leq i\leq n)$, and $\overline{V}$ consists of $DC$. When cutting node $x^{t_i}\ (1\leq i\leq k)$, $x_{in}^{t_i}$ and $x_{out}^{t_i}$ are included to $V$ or $\overline{V}$ based on the following process.

\noindent\textbf{Step 1:} When considering node $x^{t_1}$, the first topologically ordered node, there are two possible cases.
\begin{itemize}
  \item If $x_{in}^{t_1}\in V$, the edge $x_{in}^{t_1}\rightarrow x_{out}^{t_1}$ must be in $\mathcal{U}$ and the first part-cut value is $\alpha$.
  \item If $x_{in}^{t_1}\in \overline{V}$, since there are $d=d_I+d_C$ incoming edges for node $x_{in}^{t_1}$, the topologically first node in $\overline{V}$, all the $d$ edges must be contained by $\mathcal{U}$, consisting of $d_I$ and $d_C$ edges from intra-cluster and cross-cluster nodes respectively. In this case, the part-cut value is $d_I\beta_I+d_C\beta_C$.
\end{itemize}

\noindent\textbf{Step 2:} Now consider node $x^{t_i}(1\leq i\leq k)$:
\begin{itemize}
  \item If $x_{in}^{t_i}\in V$, edge $x_{in}^{t_i}\rightarrow x_{out}^{t_i}$ should be in $\mathcal{U}$.
  \item If $x_{in}^{t_i}\in \overline{V}$, the $d=d_I+d_C$ incoming edges of node $x_{in}^{t_i}$ consist of edges from $V$ and $\overline{V}$ respectively. And only the edges from $V$ are contained by $\mathcal{U}$.

      \noindent $\circ$ If $x^{t_i}$ is in a cluster, among the incoming edges from $V$, we use $a_i$ and $b_i$ to represent the number of edges \textbf{from intra-cluster} and  \textbf{cross-cluster} nodes respectively. Obviously, we have $0\leq a_i\leq d_I$ and $0\leq b_i\leq d_C$. As $x^{t_i}$ connect to all the former $\{x^{t_j}\}_{j=1}^{i-1}$ nodes (introduced in Subsection~\ref{subsec_TermDef}), when $i$ increases by $1$, either $a_i$ or $b_i$ will decrease by $1$ until either one reduced to $0$, depending on the locations of $\{x^{t_j}\}_{j=1}^{i}$.

      \noindent $\circ$ If $x^{t_i}$ is separate, $c_i$ represents the number of incoming edges from $V$. Since $d\geq k$ and $1\leq i\leq k$, it is obvious that
    \begin{equation}\label{equ_ci}
        c_i=d-(i-1).
    \end{equation}
    For given $\textbf{s}$ and $\bm{\pi}$, the $i$-th \textbf{part incoming weight} is defined as
    \begin{small}
    \begin{equation} \label{equ_wi}
        w_i(\textbf{s},\bm{\pi})\triangleq\begin{cases}
                    a_i(\textbf{s},\bm{\pi})\beta_I+b_i(\textbf{s},\bm{\pi})\beta_C, \ &\text{node }i\text{ is in a cluster},\\
                    c_i(\textbf{s},\bm{\pi})\beta_C,\ &\text{node }i\text{ is separate}
               \end{cases}.
    \end{equation}\end{small}
    When $\textbf{s}$ is fixed, we also write $w_i(\textbf{s},\bm{\pi})$ as $w_i(\bm{\pi})$ for simplicity. Meanwhile, we also use $a_i(\bm{\pi}), b_i({\bm{\pi}}), c_i({\bm{\pi}})$ for specific $\bm{\pi}$.

\end{itemize}

Subsequently, the min-cut for the given $\textbf{s}$ and  $\bm{\pi}$ is obtained as
\begin{equation}\label{equ_MC}
MC(\textbf{s},\bm{\pi})=\sum_{i=1}^k\min\{\alpha, w_i(\textbf{s},\bm{\pi})\}.
\end{equation}
The capacity of a given CSN-DSS model can be obtained by comparing the min-cuts of IFGs corresponding to selected node distributions $\textbf{s}\in \mathcal{S}$ and cluster orders $\bm{\pi}\in \Pi(\textbf{s})$. In Section~\ref{sec:NoSN}, we will sketch the main results of the capacity of the cluster DSS model, analysed in \cite{Wangjz2018} in detail.

\section{The capacity of the cluster DSS model}\label{sec:NoSN}

In \cite{Wangjz2018}, we investigated the capacity of the cluster DSS model. The main results are sketched in this section, which are generalized to the CSN-DSS model in Section~\ref{sec:withS} and compared with CSN-DSSs in Section~\ref{sec_tradoff} and \ref{sec_compare}. For a cluster DSS model with $(n,k,L,R,E=0)$, the min-cuts of all possible IFGs are compared in two steps on $\bm{\pi}$ and $\textbf{s}$, corresponding to the \textbf{vertical order algorithm} and \textbf{horizontal selection algorithm} respectively.
\begin{itemize}
  \item \textbf{Step 1:} Fix the selected node distribution $\textbf{s}$, we analyse the min-cuts for various $\bm{\pi}\in \Pi(\textbf{s})$ in Proposition~\ref{prop_MC}. The vertical order algorithm is named accordingly because the values of cluster order is generated among the clusters alternately, seeming vertically generated in Figure~\ref{fig:algPiandS}.

  \item \textbf{Step 2:} Fix the cluster order generating algorithm, we analyse the min-cuts for different $\textbf{s}$ in Proposition~\ref{prop_MC2}.
      The horizontal selection algorithm is named accordingly because the output values is generated by selecting the nodes in one cluster until no nodes left. Then select from the next cluster until there are $k$ selected nodes, see Figure~\ref{fig:algPiandS}.
\end{itemize}

\subsection{Vertical order algorithm for $d_I=R-1$}\label{subsec:vertical}

For a cluster DSS model with $(n,k,L,R,E=0)$, Proposition~\ref{prop_MC} specifies the cluster order $\bm{\pi}$ minimizing the min-cut $MC(\textbf{s},\bm{\pi})$ for an arbitrary selected node distribution $\textbf{s}$.  As introduced in \cite{Wangjz2018}, the authors of \cite{sohn2017capacity,sohn2018capacity} proposed this algorithm under the assumption that the number of helper nodes is $n-1$, namely, $d_I=R-1$ and $d_C=n-R$, which maximizes the system capacity as proven in\cite{sohn2018capacity}. However, this assumption leads to high reconstruction read cost (all alive nodes are used), which was defined in \cite{Huang2013} as the number of helper nodes to recover a failed one. We analysed the algorithm in more general cases with new methods in \cite{Wangjz2018}. The number of cross-cluster helper nodes, $d_C$, satisfies
\begin{align}\label{equ_dC}
k-R+1\leq d_C\leq n-R
\end{align}
and does not have to be $n-R$ of the constraint in \cite{sohn2017capacity}, which follows from the condition that $ k\leq d_I+d_C \leq n-1$ and $d_I=R-1$.


\begin{algorithm}
\caption{Vertical order algorithm in the CSN-DSS model}\label{alg_vs}
\begin{algorithmic}[1]
\Require Selected node distribution $\textbf{s}=(s_0,s_1,...,s_L).$ 
\Ensure Cluster order $\bm{\pi}^*=(\pi_1^*,...,\pi_k^*).$\\
Initial cluster label $j \gets 1;$
\For{$i=1$ to k}
    \If {\text{the $i$-th selected node is a separate node}}
        \State $\pi^*_i\gets 0;$ $continue;$
    \EndIf
    \If{$s_j=0$}
    \State $j=1;$
    \Else
    \State $\pi_i^*\gets j;$ $s_j\gets s_i-1;$ $j\gets (j\mod L)+1;$  \EndIf
\EndFor
\end{algorithmic}
\end{algorithm}

In Figure~\ref{fig:algPiandS}, the cluster order $\bm{\pi}^*=(1,2,1,2,1,0,1)$ is assigned vertically from the first to the fourth column, as the selected node number shows. For given $\textbf{s}$ and $\bm{\pi}$, $MC(\textbf{s},\bm{\pi})$ is obtained by (\ref{equ_MC}), which is calculated with
$$w_i(\bm{\pi})=a_i(\bm{\pi})\beta_I+b_i(\bm{\pi})\beta_C \ (1\leq i\leq k).$$
We proof a property for $a_i(\bm{\pi})$ (the coefficient of $\beta_I$) in Lemma~\ref{lemma_ai}, which can also be used to calculate $a_i(\bm{\pi})$.

\begin{figure}[!t]
    \centering
    \includegraphics[width=0.3\textwidth]{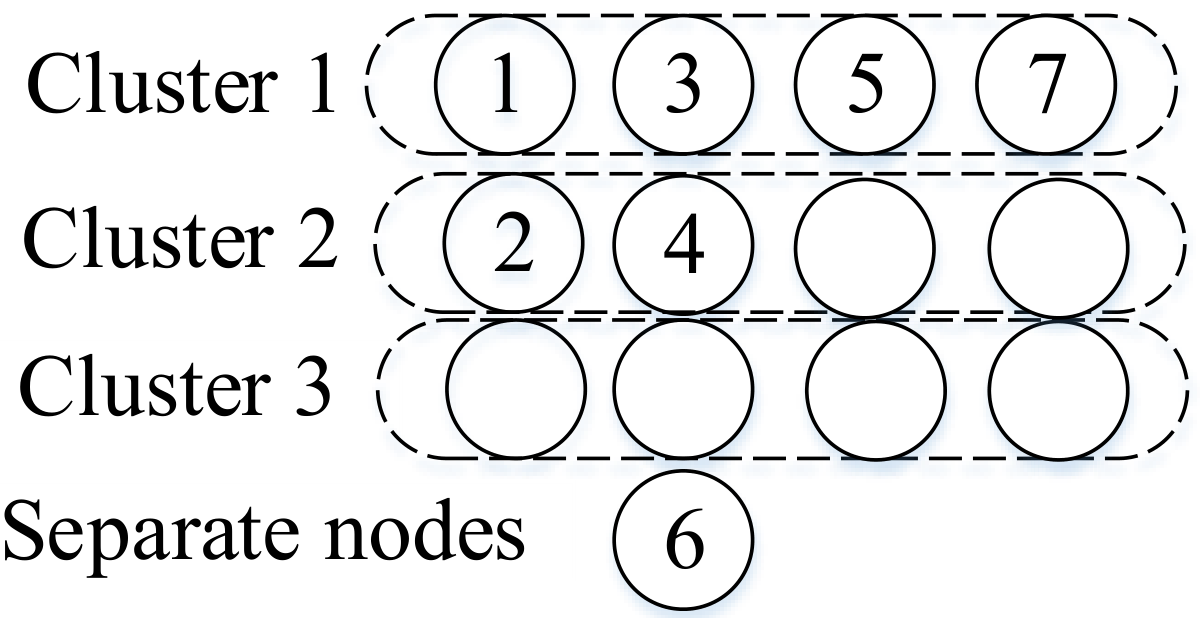}
    \caption{The cluster order $\bm{\pi}^*=(1,2,1,2,1,0,1)$ and selected node distribution $\textbf{s}^*=(1,4,2,1)$ are generated by Algorithm~1 and 2.}\label{fig:algPiandS}
\end{figure}

\begin{lemma}[\cite{Wangjz2018}]\label{lemma_ai}
For a cluster DSS model, when the selected node distribution $\textbf{s}=(0,s_1,...,s_L)$ is fixed, the elements of multi-set\footnote{The multi-set is a general definition of set, allowing multiple instances of its elements.} $[a_i(\bm{\pi})]_{i=1}^k=[a_1(\bm{\pi}),..., a_k(\bm{\pi})]$ for each cluster order $\bm{\pi}\in\Pi(\textbf{s})$ are the same. Additionally, $a_i(\bm{\pi})=d_I+1-h_{\bm{\pi}}(i)$ for $1\leq i\leq k$.
\end{lemma}

\begin{proposition}[\cite{Wangjz2018}]\label{prop_MC}
For a cluster DSS model and any fixed $\textbf{s}\in \mathcal{S}$, the cluster order $\bm{\pi}^*$ generated by Algorithm~1 minimizes the min-cut, meaning that
$$MC(\textbf{s},\bm{\pi}^*)\leq MC(\textbf{s},\bm{\pi})$$
holds for any $\bm{\pi}\in\Pi(\textbf{s})$.
\end{proposition}

We use
\begin{equation}\label{equ_pis}
 \bm{\pi}^*(\textbf{s})=(\pi^*(\textbf{s})_1,\pi^*(\textbf{s})_2,...,\pi^*(\textbf{s})_k)
\end{equation}
to represent \textbf{the output cluster order} of Algorithm~1 for any input $\textbf{s}\in \mathcal{S}$. Subsequently, Subsection~\ref{subsec:horizontal} analyse the min-cuts corresponding to $\bm{\pi}^*(\textbf{s})$ for various $\textbf{s}\in \mathcal{S}$.

\subsection{Horizontal selection algorithm for $d_I=R-1$}\label{subsec:horizontal}

For the selected node distributions $\textbf{s}\in \mathcal{S}$, we compare the min-cuts corresponding to cluster orders $\bm{\pi}^*(\textbf{s})$
and prove that $\bm{\pi}^*(\textbf{s}^*)$ minimizes the min-cut in Proposition~\ref{prop_MC2}, where $\textbf{s}^*=(s_0^*,...,s_L^*)$ is obtained by Algorithm~2. As introduced before, the cluster order generating algorithm is fixed for different $\textbf{s}$.


\begin{algorithm}
\caption{Horizontal selection algorithm in the CSN-DSS model}
\label{algo:horizontal}
\begin{algorithmic}[1] 
\Require Node parameters: $(n,k,L,R,E)$.
\Ensure Selected node distribution $\textbf{s}^*=(s_0^*, s_1^*,...,s_L^*)$. 
\State $s_0^* \gets E$
\For{$i = 1 \to L$}
    \If{$i \leq \left\lfloor \frac{k-s_0^*}{R}\right\rfloor$}
     $s_i^* \gets R$
    \EndIf
    \If{$i = \left\lfloor \frac{k-s_0^*}{R}\right\rfloor+1$}
     $s_i^* \gets k-s_0^*-\left\lfloor \frac{k-s_0^*}{R}\right\rfloor R$
    \EndIf
    \If{$i>\left\lfloor \frac{k-s_0^*}{R}\right\rfloor+1$}
     $s_i^*\gets 0$
    \EndIf
  \EndFor
\end{algorithmic}
\end{algorithm}

As introduced in \cite{Wangjz2018}, when $s_0^*=E=0$, this algorithm reduces to the cluster version ($d=n-1$)proposed in \cite{sohn2017capacity}. We also analyse this algorithm under more general assumptions as introduced in Subsection~\ref{subsec:vertical}. Additionally, we generalize the horizontal selection algorithm to the CSN-DSS model with $s_0^*=1$, which is analysed in Section~\ref{sec:withS}. Figure~\ref{fig:algPiandS} shows an example for Algorithm~2, where $k=7$, $R=4$ and $s_0^*=1$, thus $s_1^*=4$, $s_2^*=k-s_0^*-\left\lfloor \frac{k-s_0^*}{R}\right\rfloor R=6-4=2$. Subsequently, a property of $a_i(\bm{\pi}^*(\textbf{s}^*))$ and $b_i(\bm{\pi}^*(\textbf{s}^*))$, the coefficients of $\beta_I$ and $\beta_C$, is proved in Lemma~\ref{lemma_aibi}.

\begin{lemma}[\cite{Wangjz2018}]\label{lemma_aibi}
For a cluster DSS model, the coefficients of $\beta_I$ and $\beta_C$ satisfy
\begin{align}\label{equ_aibi*}
a_i(\bm{\pi}^*(\textbf{s}^*))+b_i(\bm{\pi}^*(\textbf{s}^*))=d_I+d_C+1-i,
\end{align}
for $1\leq i \leq k$, where $\textbf{s}^*$ is generated by Algorithm~2, and $\bm{\pi}^*(\cdot)$ is defined by (\ref{equ_pis}).
\end{lemma}

\begin{proposition}[\cite{Wangjz2018}]\label{prop_MC2}
For a cluster DSS model, cluster order $\bm{\pi}^*(\textbf{s}^*)$ minimizes the min-cut, meaning that
\begin{equation*}
  MC(\textbf{s}^*, \bm{\pi}^*(\textbf{s}^*))\leq MC(\textbf{s}, \bm{\pi}^*(\textbf{s}))
\end{equation*}
holds for all $\textbf{s}\in \mathcal{S}$ with $s_0=0$, where $\textbf{s}^*$ is generated by Algorithm~2 and $\bm{\pi}^*(\cdot)$ is defined by (\ref{equ_pis}).
\end{proposition}

As introduced in Subsection~\ref{subsec_PM}, the capacity of cluster DSSs is the minimum min-cut which is achieved by $\bm{\pi}^*(\textbf{s}^*)$ and $\textbf{s}^*$ generated by Algorithm~1 and 2. Note that, the algorithms are also applicative to analyse the capacity of CSN-DSSs with one separate node, which will be investigated in the following Section~\ref{sec:withS}.

\section{The capacity of CSN-DSSs}\label{sec:withS}

As introduced in Section~\ref{sec:NoSN}, the capacity of the cluster DSS model is achieved by cluster order $\bm{\pi}^*(\textbf{s}^*)$ generated with Algorithm~1, where  $\textbf{s}^*$ is generated by Algorithm~2. In this section, we first analyse the min-cuts corresponding to cluster orders with one separate node in Theorem~\ref{theorem_MC3}, where the applicability of Algorithm~1 and 2 is proved. Subsequently, when the location of the added separate node varies, the corresponding min-cuts are analysed in Theorem~\ref{thm_MCj}, with which the capacity of the CSN-DSS model is formulated in Theorem~\ref{theorem_MCSN} finally.

In Theorem~\ref{theorem_MC3}, we assume the separate node is at any given location $j$ ($1\leq j\leq k$) in every cluster orders and compare the min-cuts combining the two aspects corresponding to Proposition~\ref{prop_MC} and Proposition~\ref{prop_MC2}.

\begin{theorem}\label{theorem_MC3}
  For a CSN-DSS model with the separate selected node at location $j$ ($1\leq j\leq k$), the cluster order $\bm{\pi}^*(\textbf{s}^*)$ minimizes the min-cut, meaning that
  \begin{equation}
    MC(\textbf{s}^*, \bm{\pi}^*(\textbf{s}^*))\leq MC(\textbf{s}, \bm{\pi}),
  \end{equation}
  holds for all $\textbf{s}\in \mathcal{S}$ with $s_0=1$ and $\bm{\pi}\in \Pi(\textbf{s})$ with $\pi_j=0$, where $\textbf{s}^*$ is generated by Algorithm~2 and $\bm{\pi}^*(\cdot)$ is defined by (\ref{equ_pis}).
\end{theorem}

\begin{proof}
In \cite{Wangjz2018}, we sketch the main idea of this proof, which will be completed here. We will reduce this proof to Proposition~\ref{prop_MC} and Proposition~\ref{prop_MC2}. As the $j$-th selected node is separate and fixed for each $\bm{\pi}$, we only need to consider the part of selected cluster nodes by analysing the influence of adding a separate selected node. It is convenient to represent the cluster order $\bm{\pi}$ with another cluster order $\overline{\bm{\pi}}$ without separate nodes, as formula (\ref{equ_newpi}) shows (see Figure~\ref{fig_SC} (b), (c)). The part incoming weights $w_i(\bm{\pi})$ $(1\leq i\leq k)$ are then expressed with $w_i(\overline{\bm{\pi}})$, and this theorem is proved by analysing $w_i(\overline{\bm{\pi}})$ with similar methods used in Proposition~\ref{prop_MC} and \ref{prop_MC2}.
See \ref{app_thm_MC3} for more details.
\end{proof}

\begin{rmk}
After adding a separate selected node, other $k-1$ selected cluster nodes generated by Algorithm~1 and 2 are ordered similarly to the situation without separate nodes. It can be verified that Lemma~\ref{lemma_ai} and Lemma~\ref{lemma_aibi} still hold for the cluster selected nodes in the CSN-DSSs with one separate node. The added separate node will not influence the properties of the other $k-1$ cluster selected nodes, which can be proved with the same methods to Lemma~\ref{lemma_ai} and Lemma~\ref{lemma_aibi}. We omit these proofs due to space limitation and use Lemma~\ref{lemma_ai} and Lemma~\ref{lemma_aibi} directly.
\end{rmk}

In order to analyse the relationship between the min-cuts and the separate selected node location, let $MC^*_j$ denote the min-cut corresponding to $\textbf{s}^*$ and $\bm{\pi}^*(\textbf{s}^*)$ mentioned in Theorem~\ref{theorem_MC3}, where $s^*_0=1$ and $\pi^*(\textbf{s}^*)_j=0$, namely,
\begin{align}\label{equ_MCj}
  MC_j^*\triangleq MC(\textbf{s}^*, \bm{\pi}^*(\textbf{s}^*))
\end{align}
for $1\leq j\leq k$. The corresponding cluster order is represented with $\bm{\pi}^{(j)}$, namely,
\begin{align}\label{equ_pij}
\bm{\pi}^{(j)}&\triangleq\bm{\pi}^*(\textbf{s}^*)\text{ with }\pi^*(\textbf{s}^*)_j=0.
\end{align}
\begin{rmk}\label{rmk_pi} Note that the separate selected node location is not identified by $\bm{\pi}^*(\textbf{s}^*)$ explicitly. For notational simplicity, we use $\bm{\pi}^*(\textbf{s}^*)$ to denote the cluster order without separate nodes if not explicitly state. On the other hand, let $\bm{\pi}^{(j)}$ denote the cluster order with one separate selected node at location $j$. Both $\bm{\pi}^*(\textbf{s}^*)$ and $\bm{\pi}^{(j)}$ are generated by Algorithm~1 and 2.
\end{rmk}

Subsequently, we compare min-cuts $MC_j^*$ and $MC_{j+1}^*$ and prove that $MC_j^*\geq MC_{j+1}^*$ for $1\leq j\leq k-1$ in the following Theorem~\ref{thm_MCj}, with which the final capacity of the CSN-DSS model is derived in Theorem~\ref{theorem_MCSN}.

\begin{theorem}\label{thm_MCj}
  For a CSN-DSS model with one separate node, the min-cuts of IFGs corresponding to the cluster order $\bm{\pi}^{(j)}$ satisfy that
  $$MC_j^*\geq MC_{j+1}^*,$$
  for $1\leq j\leq k-1$, where $\bm{\pi}^{(j)}$ is defined by (\ref{equ_pij}) with the separate node at location $j$ and $MC_j^*$ is defined by (\ref{equ_MCj}).
\end{theorem}
\begin{proof}
To prove this theorem, the part incoming weights $w_i\big(\bm{\pi}^{(j)}\big)$ and $w_i\big(\bm{\pi}^{(j+1)}\big)$ ($1\leq i \leq k$) are compared one by one. By analysing cluster orders $\bm{\pi}^{(j)}$ and $\bm{\pi}^{(j+1)}$, we find that $$w_i\big(\bm{\pi}^{(j)}\big)=w_i\big(\bm{\pi}^{(j+1)}\big)$$ for $i\in [k]\setminus \{j,j+1\}$\footnote{Let [k] denote the integer set $\{1,2,...k\}$.}. Hence, we only need to compare $w_j\big(\bm{\pi}^{(j)}\big)$, $w_{j+1}\big(\bm{\pi}^{(j)}\big)$ and $w_j\big(\bm{\pi}^{(j+1)}\big)$, $w_{j+1}\big(\bm{\pi}^{(j+1)}\big)$. We enumerate all the possible cases of the above four components to complete this proof. See \ref{app_lem_MCj} for more details.
\end{proof}

\begin{figure*}[!t]
\center
  \subfloat[$\bm{\pi}^*(\textbf{s}^*)={(1,2,1,2,1,2,1)}$]{\includegraphics[width=0.27\textwidth]{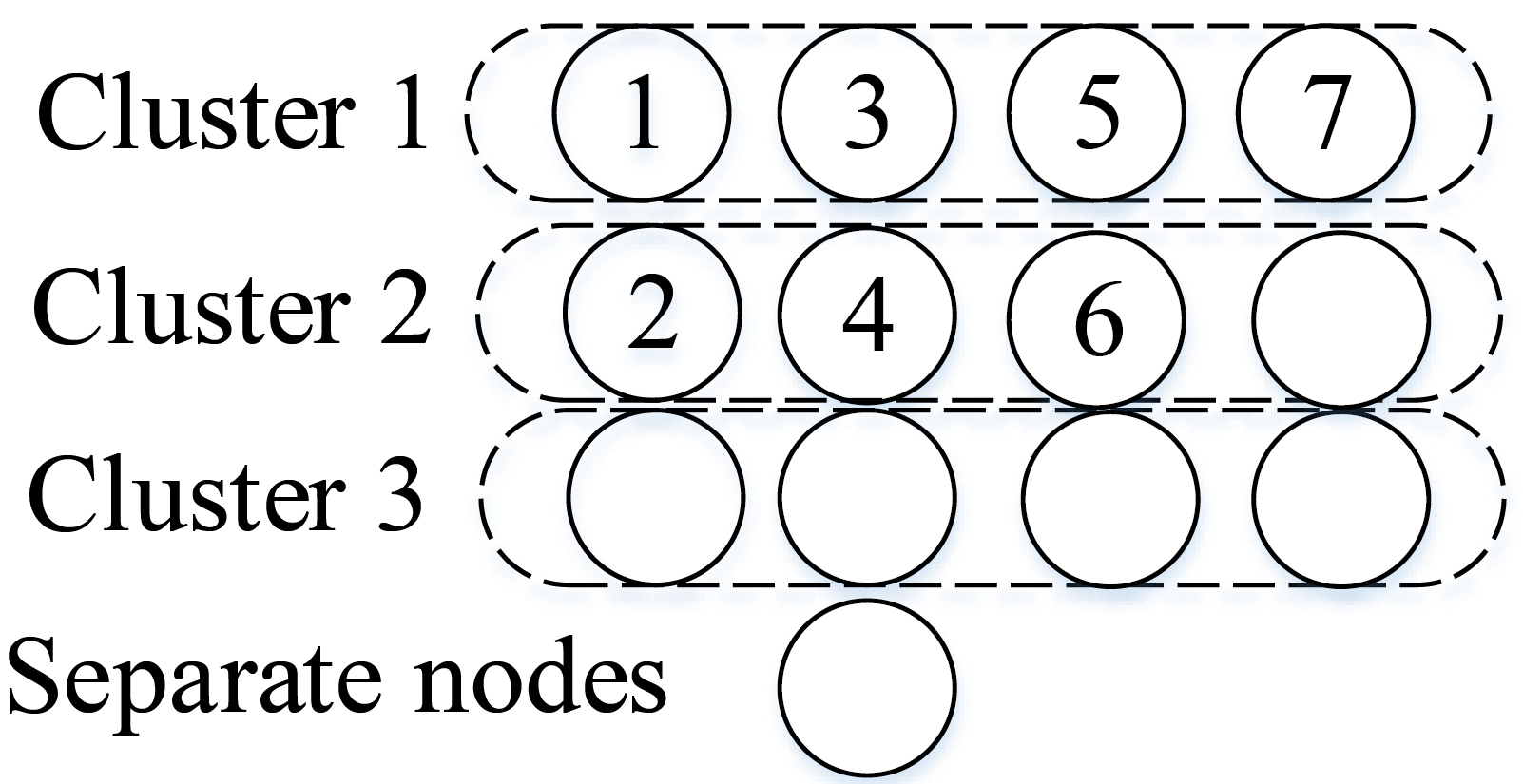}}
  \hspace{0.02\textwidth}
  \subfloat[$\bm{\pi}^{(6)}={(1,2,1,2,1,0,1)}$]{\includegraphics[width=0.27\textwidth]{OneSeparateCluster2.png}}
  \hspace{0.02\textwidth}
  \subfloat[$\bm{\pi}^{(7)}={(1,2,1,2,1,1,0)}$]{\includegraphics[width=0.27\textwidth]{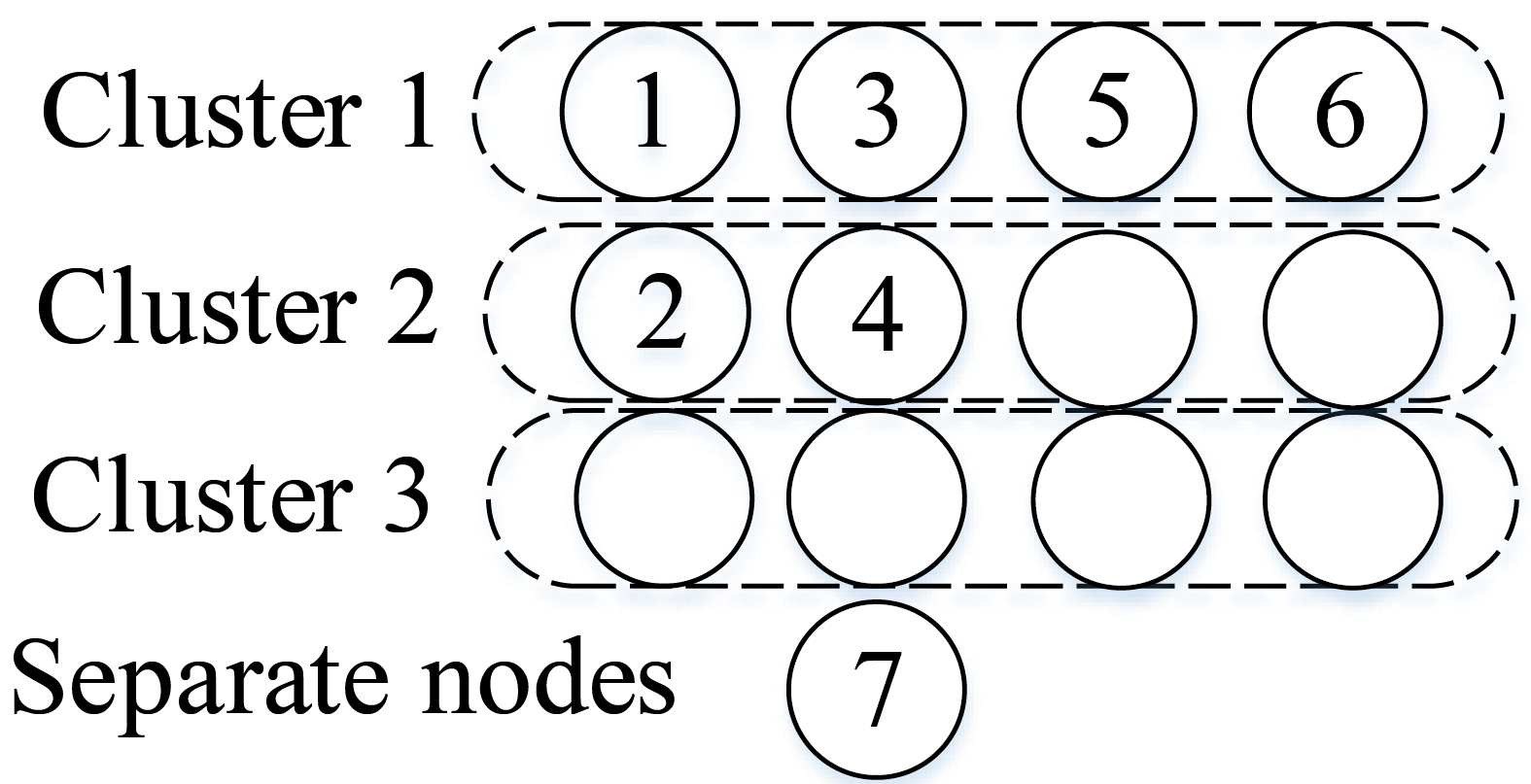}}
  \caption{In (a), the cluster order is $\bm{\pi}^*(\textbf{s}^*)={(1,2,1,2,1,2,1)}$ for $\textbf{s}^*=(0,4,3)$. In (b) and (c), the cluster orders $\bm{\pi}^{(6)}$ and $\bm{\pi}^{(7)}$ are corresponding to $\textbf{s}=(1,4,3,0)$, where the separate selected node locations are $6$ and $7$, respectively.}\label{fig_theoremOneSep}
\end{figure*}

With Theorem~\ref{thm_MCj}, the minimum min-cut with one separate selected node is derived. We obtain the capacity of the CSN-DSS model in Theorem~\ref{theorem_MCSN} through comparing $MC_j^*$ and $MC(\textbf{s}^*,\bm{\pi}^*(\textbf{s}^*))$, where $\bm{\pi}^*(\textbf{s}^*)$ without the separate node is generated by Algorithm~1 and 2, as introduced in Remark~\ref{rmk_pi}.

\begin{theorem} \label{theorem_MCSN}
For a CSN-DSS model with one separate node, the system capacity is
\begin{align*}
MC_{k}^*=\sum_{i=1}^{k}\min\big\{w_i\big(\bm{\pi}^{(k)}\big), \alpha\big\},
\end{align*}
where $\bm{\pi}^{(k)}$ is defined by (\ref{equ_pij}) .
\end{theorem}
\begin{proof}
As proven in Theorem~\ref{thm_MCj}, $MC_{k}^*$ is the minimum min-cut with one separate selected node.
Proposition~\ref{prop_MC2} proves that $MC(\textbf{s}^*,\bm{\pi}^*(\textbf{s}^*))$ is the minimum min-cut corresponding to cluster orders without separate nodes. Hence, we only need to compare $MC_{k}^*$ and $MC(\textbf{s}^*,\bm{\pi}^*(\textbf{s}^*))$.

For the cluster order $\bm{\pi}^*(\textbf{s}^*)$ without separate nodes, we can always find a cluster order $\bm{\pi}^{(j)}$ with one fixed separate selected node at location $j$, satisfying that
\begin{align}\label{equ_pit}
\pi^*(\textbf{s}^*)_t=\pi^{(j)}_t,
\end{align}
for $1\leq t\leq k \text{ and } t\neq j$, where $j=k-\left\lfloor \frac{k}{R}\right\rfloor$. Then
\begin{align}
w_t(\bm{\pi}^*(\textbf{s}^*))=w_t\big(\bm{\pi}^{(j)}\big),
\end{align}
for $1\leq t\leq k \text{ and } t\neq j$. When $t=j$,
\begin{align*}
w_j(\bm{\pi}^*(\textbf{s}^*))&=a_j(\bm{\pi}^*(\textbf{s}^*))\beta_I+b_j(\bm{\pi}^*(\textbf{s}^*))\beta_C\\
&\overset{(a)}{\geq} \big(a_j(\bm{\pi}^*(\textbf{s}^*))+b_j(\bm{\pi}^*(\textbf{s}^*))\big)\beta_C\\
&\overset{(b)}{=}(d_I+d_C+1-j)\beta_C=w_j\big(\bm{\pi}^{(j)}\big),
\end{align*}
where (a) is because of $\beta_I\geq \beta_C$ and (b) is based on Lemma~\ref{lemma_aibi}. Hence,
\begin{align*}
MC(\textbf{s}^*,\bm{\pi}^*(\textbf{s}^*))\geq MC_j^*.
\end{align*}
Then
$MC_k^* \leq MC_j^* \leq MC(\textbf{s}^*,\bm{\pi}^*(\textbf{s}^*)).$

As shown in Figure~\ref{fig_theoremOneSep}, for the cluster order $\bm{\pi}^*(\textbf{s}^*)=(1,2,1,2,1,2,1)$ in (a), we can find a cluster order $\bm{\pi}^{(6)}={(1,2,1,2,1,0,1)}$ in (b) with the separate node at location $j=k-\left\lfloor \frac{k}{R}\right\rfloor=7-\left\lfloor \frac{7}{4}\right\rfloor=6$, satisfying condition (\ref{equ_pit}). Additionally, $w_6(\bm{\pi}^*(\textbf{s}^*))=(d_I-2)\beta_I+(d_C-3)\beta_C \geq (d_I+d_C-5)\beta_C =w_6\big(\bm{\pi}^{(6)}\big)$. Hence, $MC_6^* \leq MC(\textbf{s}^*, \bm{\pi}^*(\textbf{s}^*))$. In Figure~\ref{fig_theoremOneSep} (c), the locations of separate selected nods is $k=7$. Based on Theorem~\ref{thm_MCj}, $MC_7^*\leq MC_6^*$. Then $MC_7^*\leq MC_6^* \leq MC(\textbf{s}^*,\bm{\pi}^*(\textbf{s}^*))$.
\end{proof}

\section{Tradeoffs for the CSN-DSS and cluster DSS model}\label{sec_tradoff}

In Theorem~\ref{theorem_MCSN}, the capacity of a CSN-DSS model with one separate node is specified as $MC_k^*$, the left part of formula (\ref{equ_finalboundS}). As introduced in Subsection~\ref{subsec_PM}, by analysing
\begin{align}\label{equ_finalboundS}
\sum_{i=1}^{k}\min\big\{w_i\big(\bm{\pi}^{(k)}\big), \alpha\big\}\geq \mathcal{M},
\end{align}\noindent
the tradeoff bound between node storage $\alpha$ and repair bandwidth parameters $(d_I, \beta_I, d_C, \beta_C)$ will be characterized, which is figured out in formula (\ref{equ_CSNalpha}) and calculated with formula (\ref{equ_wn*}).

Let $w_i^* \triangleq w_i\big(\bm{\pi}^{(k)}\big)$ for simplicity and $w_i^*$ can be figured out as
\begin{small}\begin{align}\label{equ_wn*}
w_{k-i+1}^*=\begin{cases}
       \left(R-\left\lceil\frac{i}{\left\lfloor\frac{k-1}{R}\right\rfloor+1}\right\rceil\right)\beta_I+\left(d_C-i+\left\lceil \frac{i}{\left\lfloor\frac{k-1}{R}\right\rfloor+1}\right\rceil\right)\beta_C,\\\\
        \qquad\qquad\qquad 1\leq i \leq \left(\left\lfloor\frac{k-1}{R}\right\rfloor+1\right)\left(k-1-\left\lfloor\frac{k-1}{R}\right\rfloor R\right)\\
       \\
       \left(\left\lfloor\frac{k-i-1}{\left\lfloor\frac{k-1}{R}\right\rfloor}\right\rfloor\right)\beta_I+\left(d_C+R-i-\left\lfloor\frac{k-i-1}{\left\lfloor\frac{k-1}{R}\right\rfloor}\right\rfloor\right)\beta_C, \\\\
       \qquad\qquad\ \left(\left\lfloor\frac{k-1}{R}\right\rfloor+1\right)\left(k-1-\left\lfloor\frac{k-1}{R}\right\rfloor R\right)<i\leq k-1\\ \\
       (R+d_C-k)\beta_C, \qquad\qquad\qquad i=k
      \end{cases}.
\end{align}\end{small}\noindent
In formula (\ref{equ_wn*}), the subscript of $w^*$ is set as $k-i+1$ for convenience to ensure
$$w^*_1\leq w^*_2\leq ... \leq w^*_k$$
which will not change $MC_{k}^*$. In fact, we can get $w^*_2\leq ... \leq w^*_k$ by using Algorithm~1 and 2. In addition, because of Lemma~\ref{lemma_aibi} and $\beta_I\geq \beta_C$, we can also find that $w_1^*=(d_I+d_C+1-k)\beta_C$ is the minimum among $w_i^*(1\leq i\leq k)$.

Assume $\alpha^*$ is the minimum $\alpha$ satisfying (\ref{equ_finalboundS}). With similar methods in \cite{Wangjz2018}, the tradeoff bound is obtained (See Figure~\ref{fig_tradeoffDiffk} for example).
For $i=2,...,k$,
\begin{small}\begin{align}\label{equ_CSNalpha}
\alpha^*=\frac{\mathcal{M}-\sum_{j=1}^{i-1}w^*_j}{k-i+1},
\end{align}\end{small}\noindent
for $\mathcal{M}\in \big(\sum_{j=1}^{i-1}w^*_j+(k-i+1)w^*_{i-1}, \sum_{j=1}^{i} w^*_j+(k-i)w^*_i\big]$.
When $i=1$, $\alpha^*=\mathcal{M}/k$ for $\mathcal{M}\in [0,w^*_1]$.

As introduced in \cite{Wangjz2018}, the tradeoff bound of cluster DSSs holds the same expression with (\ref{equ_CSNalpha}), but the values of $w^*_j$ are found differently as
\begin{small}\begin{align}\label{equ_w*}
w_{k-i+1}^*=\begin{cases}
       \left(R-\left\lceil\frac{i}{\left\lfloor\frac{k}{R}\right\rfloor+1}\right\rceil\right)\beta_I+\left(d_C-i+\left\lceil \frac{i}{\left\lfloor\frac{k}{R}\right\rfloor+1}\right\rceil\right)\beta_C, \\\\
       \qquad\qquad\qquad 1\leq i \leq \left(\left\lfloor\frac{k}{R}\right\rfloor+1\right)\left(k-\left\lfloor\frac{k}{R}\right\rfloor R\right)\\
       \\
       \left(\left\lfloor\frac{k-i}{\left\lfloor\frac{k}{R}\right\rfloor}\right\rfloor\right)\beta_I+\left(d_C+R-i-\left\lfloor\frac{k-i}{\left\lfloor\frac{k}{R}\right\rfloor}\right\rfloor\right)\beta_C, \\\\ \qquad\qquad\qquad  \left(\left\lfloor\frac{k}{R}\right\rfloor+1\right)\left(k-\left\lfloor\frac{k}{R}\right\rfloor R\right)<i\leq k.
      \end{cases}
\end{align}\end{small}\noindent

\begin{rmk}
When $\beta_I=\beta_C$, Equation (\ref{equ_wn*}) and (\ref{equ_w*}) both reduce to $$w_{k-i+1}^*=(R+d_C-i)\beta_C$$ for $1\leq i\leq k$, according with the tradeoff bounds in \cite{Dimakis2010} without differentiating clusters and separate nodes. Note that we assume all alive nodes in the same cluster with the failed one are used, thus $d=d_I+d_C=R-1+d_C$.
\end{rmk}

Numerical comparisons of tradeoff bounds for CSN-DSSs with and without the separate node are illustrated in Figure~\ref{fig_tradeoffDiffk} in Section~\ref{sec_compare} where we also analyse the differences theoretically. Similarly to the regenerating code constructions mentioned in \cite{Survey2011,Fazeli2016}, interference alignment can also be used in the CSN-DSS with one separate node. We gave a code construction example for cluster DSSs without separate nodes in \cite{Wangjz2018}. The construction problem of CSN-DSSs with one separate node is investigated in the following Section~\ref{sec:consWithS}.

\section{Code constructions for CSN-DSSs with one separated nodes}\label{sec:consWithS}

Minimum storage regenerating (MSR) code was proposed in \cite{Dimakis2010}, indicating the code constructions which achieve the minimum storage point in the optimum tradeoff curve between storage and repair bandwidth. As MSR code is optimal in
terms of the redundancy-reliability \cite{Dimakis2010}, there are lots of research works on MSR codes \cite{Survey2011,Fazeli2016,Rashmi2011Optimal,tang2015anew,yang2015system}. In this section, we use the interference alignment scheme \cite{wu2009reducing} to give an MSR code construction. For example, the point $(\beta_C=1,\alpha=2)$ is the minimum storage point of the CSN-DSS model with $(n=5,\ k=3,\ L=2,\ R=2,\ E=1)$ in Figure~\ref{fig:consTradeoffWithS}, where the storage/bandwidth parameter constraints are $\beta_I=2\beta_C$, $d_I=1$, $d_C=3$. Under the parameter constraints of the CSN-DSS with one separate node in Figure~\ref{fig:consTradeoffWithS}, we will introduce an MSR code construction achieving the point $(\beta_C=1,\alpha=2)$.

\begin{figure}[t]
    \centering
    \includegraphics[width=0.4\textwidth]{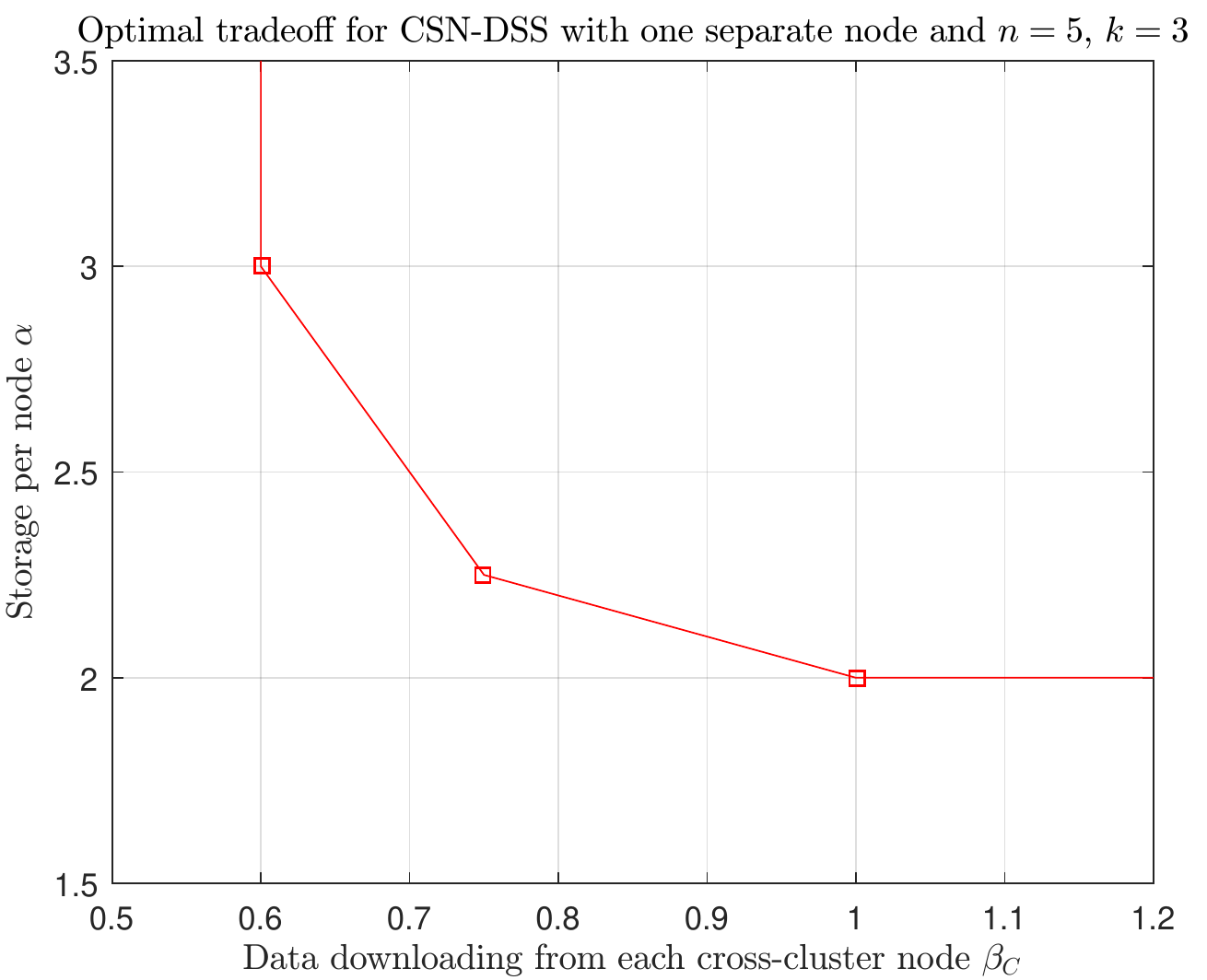}
    \caption{The optimal tradeoff curve between node storage $\alpha$ and each cross-cluster bandwidth $\beta_C$, for the CSN-DSS model, where $\beta_I=2\beta_C$, $d_I=1$, $d_C=3$ and $\mathcal{M}=6$.} \label{fig:consTradeoffWithS}
 \end{figure}

\begin{figure}[t]
    \centering
    \includegraphics[width=0.3\textwidth]{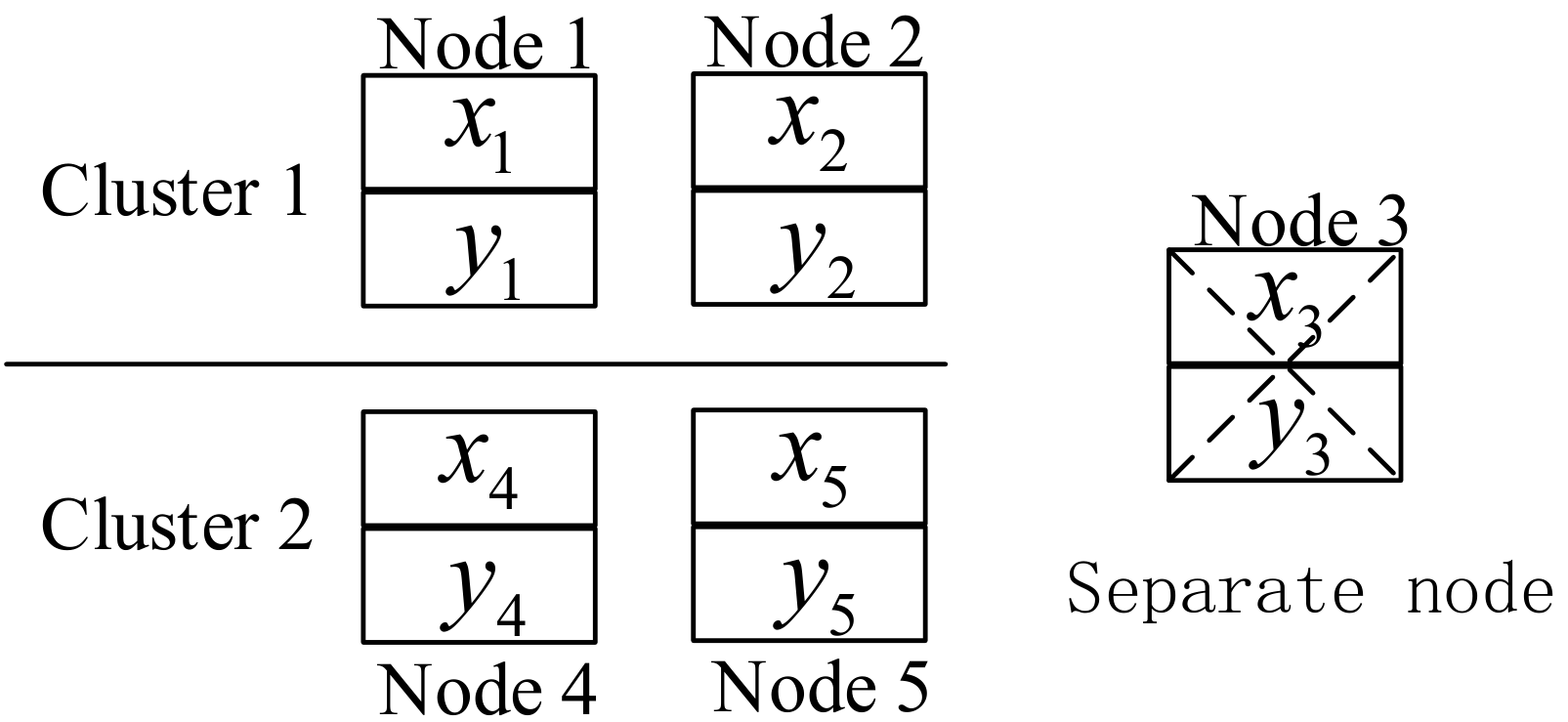}
    \caption{The code construction illustration for the CSN-DSS model, where $\alpha=2$, $\beta_I=2$, $\beta_C=1$, $d_I=1$, $d_C=3$ and $\mathcal{M}=6$.} \label{fig_consWiths}
  \end{figure}

Figure~\ref{fig_consWiths} shows the system node configurations, where there are one separate node and $2$ clusters each with $2$ nodes. We assume the original file consists of $\mathcal{M}=6$ symbols. The encoding and repair procedures are introduced as follows.

\noindent\textbf{Encoding procedure:} The original file consists of $\mathcal{M}=6$ symbols represented by $x_1$, $x_2$, $x_3$, $y_1$, $y_2$, $y_3$ and stored in node 1, 2 and 3 as Figure~\ref{fig_consWiths} shows. We use two $(5,3)$-MDS codes to encode $(x_1,x_2,x_3)$ and $(y_1,y_2,y_3)$ respectively. Let
\begin{align*}
(x_1,x_2,x_3,x_4,x_5)&=(x_1,x_2,x_3)\left[\textbf{I}_{3\times 3}|\textbf{A}\right],\\
(y_1,y_2,y_3,y_4,y_5)&=(y_1,y_2,y_3)\left[\textbf{I}_{3\times 3}|\textbf{B}\right],
\end{align*}
where $\textbf{I}_{3\times 3}$ is an identity matrix. $\textbf{A}=(a_{ij})_{3\times 2}$ and $\textbf{B}=(b_{ij})_{3\times 2}$ are the encoding matrices. Then
$$x_4=a_{11}x_1+a_{21}x_2+a_{31}x_3,\text{ } x_5=a_{12}x_1+a_{22}x_2+a_{32}x_3,$$
$$y_4=b_{11}y_1+b_{21}y_2+b_{31}y_3,\text{ } y_5=b_{12}y_1+b_{22}y_2+b_{32}y_3.$$
The construction of MSR codes is to find the proper $\textbf{A}$ and $\textbf{B}$ satisfying the storage/repair conditions.

\noindent\textbf{Repair procedure:} As constrained by the storage/repair parameters, when Node 1 has failed, the newcomer will download $\beta_I=2$ symbols from Node 2 and $\beta_C=1$ symbol from each of Node 3, 4 and 5 respectively. On the other hand, when the separate Node 3 has failed, the newcomer will download one symbol from each of the four alive nodes. We first assume the separate Node 3 has failed and the four symbols downloaded from Node 1, 2, 4 and 5 respectively are
      \begin{align*}
        &symbol_1=c_{11}x_1+c_{12}y_1 \text{ from Node 1},\\
        &symbol_2=c_{21}x_2+c_{22}y_2 \text{ from Node 2},\\
        &symbol_4=c_{41}x_4+c_{42}y_4 \text{ from Node 4},\\
        &symbol_5=c_{51}x_5+c_{52}y_5 \text{ from Node 5},
      \end{align*}
where $c_{i1}$ and $c_{i2}$ ($i=1,2,4,5$) are called \textbf{repair download parameters}, which are designed beforehand to satisfy condition $(\ref{equ:rankcon})$. Note that the subscript of $symbol$ only represents which cluster it is from. With the equations in the encoding procedure, we get
      \begin{align}
      symbol_1&=c_{11}x_1+c_{12}y_1,\label{eq:sy1}\\
      symbol_2&=c_{21}x_2+c_{22}y_2,\label{eq:sy2}\\
      symbol_4&=c_{41}a_{11}x_1+c_{42}b_{11}y_1+c_{41}a_{21}x_2+c_{42}b_{21}y_2 \notag\\
      &\ \ \ +c_{41}a_{31}x_3+c_{42}b_{31}y_3,\label{eq:sy4}\\
      symbol_5&=c_{51}a_{12}x_1+c_{52}b_{12}y_1+c_{51}a_{22}x_2+c_{52}b_{22}y_2\notag\\
      &\ \ \ +c_{51}a_{32}x_3+c_{52}b_{32}y_3.\label{eq:sy5}
      \end{align}
If the coefficients of $x_i$ and $y_i$ ($1\leq i\leq 3$) in the above 4 equations satisfy that
     \begin{small} \begin{align}\label{equ:rankcon}
      \textbf{rank}\left(\left[\begin{matrix}
             c_{11} & c_{12} \\
             c_{41}a_{11} & c_{42}b_{11}\\
             c_{51}a_{12} & c_{52}b_{12}
           \end{matrix}
      \right]\right)=1,\text{   }
      \textbf{rank}\left(\left[\begin{matrix}
             c_{21} & c_{22} \\
             c_{41}a_{21} & c_{42}b_{21}\\
             c_{51}a_{22} & c_{52}b_{22}
           \end{matrix}
      \right]\right)=1
      \end{align}\end{small}
      and
      \begin{small}
      \begin{align}\label{equ:rankcon2}
      \textbf{rank}\left(\left[\begin{matrix}
             c_{41}a_{31} & c_{42}b_{32}\\
             c_{51}a_{32} & c_{52}b_{32}
           \end{matrix}
      \right]\right)=2,
      \end{align}\end{small}\noindent
we can eliminate $x_1$, $y_1$, $x_2$, $y_2$ in Eq.~\ref{eq:sy4} and Eq.~\ref{eq:sy5} with Eq.~\ref{eq:sy1} and Eq.~\ref{eq:sy2}. Meanwhile, $x_3$ and $y_3$ are solved out.

When a cluster node has failed, a similar repair procedure can be executed. As more symbols can be downloaded from intra-cluster nodes, it may be easier to satisfy conditions (\ref{equ:rankcon}) and (\ref{equ:rankcon2}). In \cite{Survey2011}, the authors proved that there exist MDS codes and repair download parameters satisfying the condition (\ref{equ:rankcon}) and (\ref{equ:rankcon2}). The constructions of MDS codes and repair download parameters covering all possible node failures are more complicated. When the system properties of the CSN-DSS with one separate node are considered, the constructions should be easier to find, but more future works are needed.

\section{Comparison of CSN-DSSs with and without separate nodes}\label{sec_compare}

In Section~\ref{sec_tradoff}, the tradeoff bounds for CSN-DSSs with and without separate nodes are formulated based on the capacities derived in Section~\ref{sec:NoSN} and Section~\ref{sec:withS}. We will compare the capacity and tradeoff bounds of CSN-DSSs with and without the separate node in this Section~\ref{sec_compare}. The tradeoff bounds of cluster DSSs and those after adding a separate node are illustrated in Figure~\ref{fig_tradeoffDiffk}. In Theorem~\ref{theorem_compareS}, we theoretically prove that when each cluster contains $R$ nodes and any $k$ nodes suffice to reconstruct the original file, adding a separate node will keep the capacity if $R|k$, and reduce the capacity otherwise. Two examples are used to illustrate our main proof ideas. Example~\ref{exam_k9} shows the situation where the capacity is reduced. Consequently, the tradeoff bounds move left than those without the separate node (see Figure~\ref{fig_tradeoffDiffk} (b)). On the other hand, adding one separate node will not affect the system capacity in Example~\ref{exam_k8}.

\begin{figure}[t]
  \centering
  \subfloat[$\bm{\pi}^*(\textbf{s}^*)={(1,2,3,1,2,1,2,1,2)}$]{\includegraphics[width=0.23\textwidth]{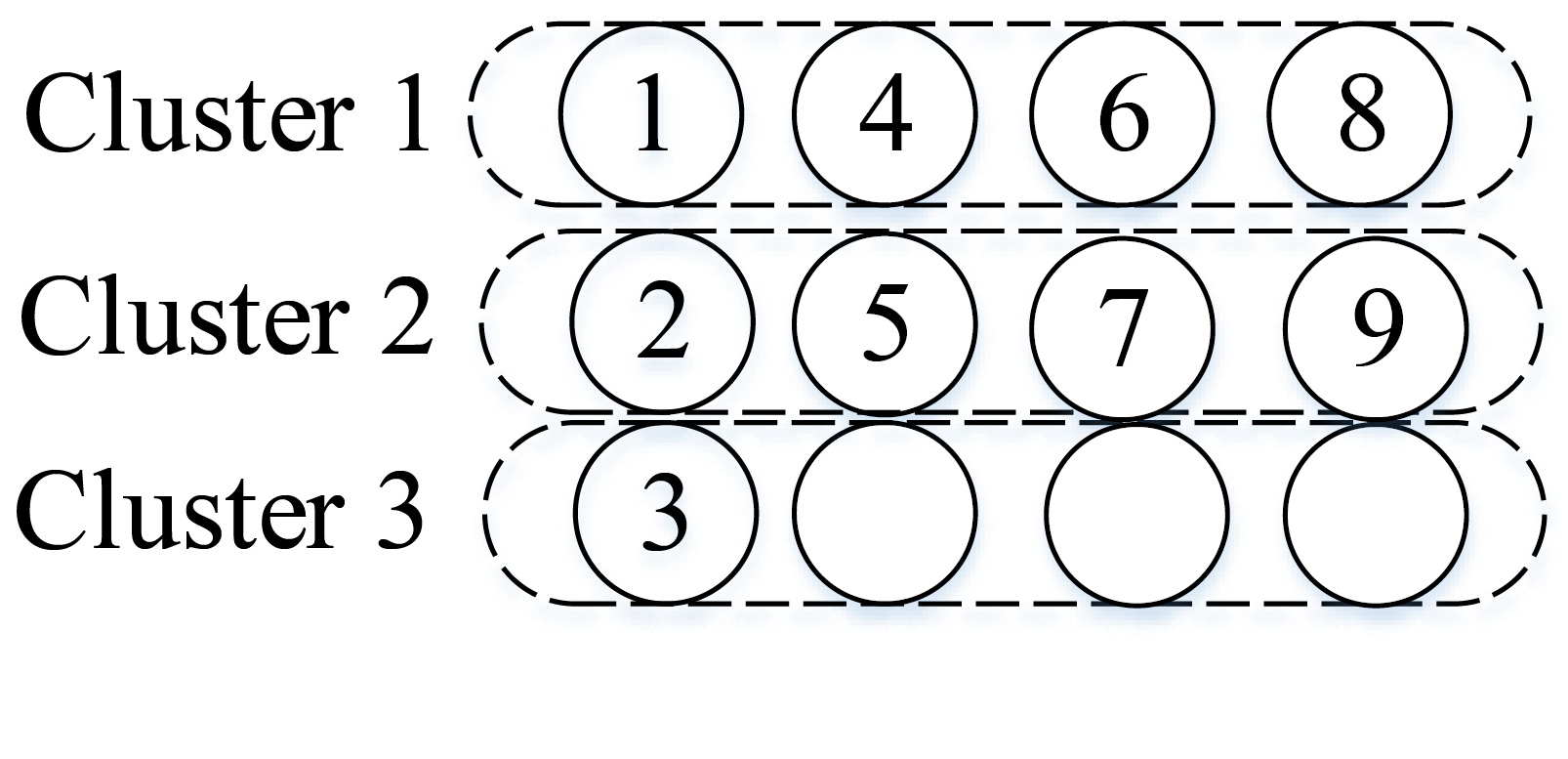}}
  \hspace{0.02\textwidth}
  \subfloat[$\bm{\pi}^{(9)}={(1,2,1,2,1,2,1,2,0)}$]{\includegraphics[width=0.23\textwidth]{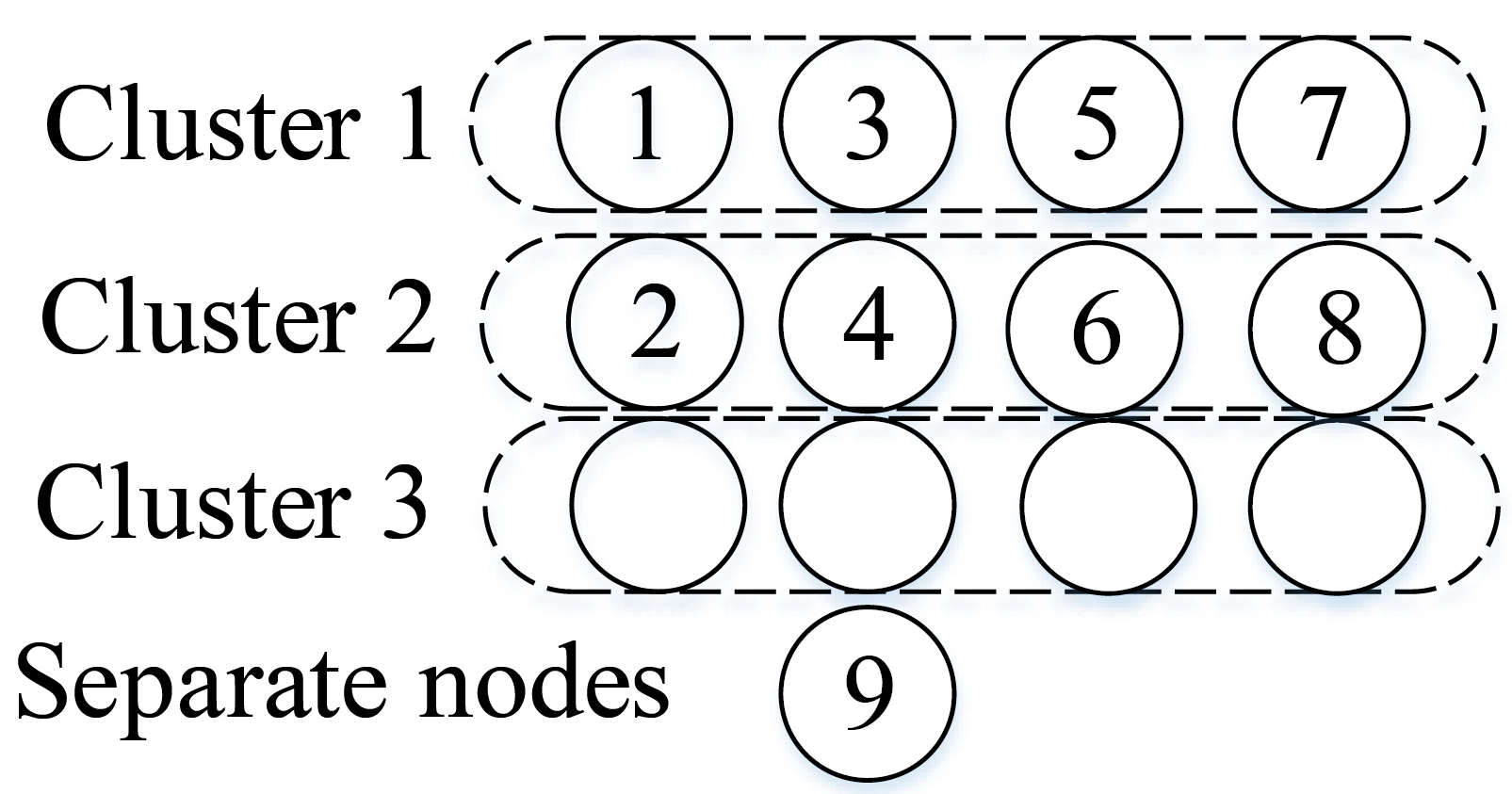}}
  \caption{In (a), the cluster order is $\bm{\pi}^*(\textbf{s}^*)$ for $\textbf{s}^*=(0,4,4,1)$, which achieves the capacity of cluster DSS model. In (b), the cluster order and selected node distribution are $\bm{\pi}^{(9)}$ and $\textbf{s}^*=(1,4,4,0)$ respectively, which achieves the capacity of the CSN-DSS model.}\label{fig:compare9}
\end{figure}

\begin{example}\label{exam_k9}
Figure~\ref{fig:compare9} shows the cluster DSS with $(n=12,\ k=9,\ L=3,\ R=4,\ E=0)$ and the CSN-DSS after adding a separate node. The new CSN-DSS model possesses the same node and storage/repair parameters except $n$ and $E$. Based on Theorem~\ref{theorem_MC3} and Theorem~\ref{theorem_MCSN}, the cluster orders $\bm{\pi}^*(\textbf{s}^*)={(1,2,3,1,2,1,2,1,2)}$ and $\bm{\pi}^{(9)}={(1,2,1,2,1,2,1,2,0)}$ respectively achieve the capacity of these two systems, namely,
\begin{align*}
MC(\textbf{s}^*,\bm{\pi}^*(\textbf{s}^*))=\sum_{i=1}^{k}\min\big\{w_i\big(\bm{\pi}^*(\textbf{s}^*)\big), \alpha\big\}
\end{align*}
and
\begin{align*}
MC_9^*=\sum_{i=1}^{k}\min\big\{w_i\big(\bm{\pi}^{(9)}\big), \alpha\big\}.
\end{align*}
With the methods of calculating $w_i(\bm{\pi})$ (see equation (\ref{equ_wi})), we compare $w_i(\bm{\pi}^*(\textbf{s}^*))$ and $w_i\big(\bm{\pi}^{(9)}\big)$ one by one as follows.
\begin{itemize}
  \item We can verify that $w_i(\bm{\pi}^*(\textbf{s}^*))=w_i\big(\bm{\pi}^{(9)}\big)$ for $i=1,2$.
  \item When $i=4$, although node 4 is in different clusters in Figure~\ref{fig:compare9} (a) and (b), we can get $a_4(\bm{\pi}^*(\textbf{s}^*))=a_4\big(\bm{\pi}^{(9)}\big)=d_I-1$ (based on Lemma~\ref{lemma_ai}) and $b_4(\bm{\pi}^*(\textbf{s}^*))=b_4\big(\bm{\pi}^{(9)}\big)=d_C-2$ (based on Lemma~\ref{lemma_aibi}). Then $w_4(\bm{\pi}^*(\textbf{s}^*))=w_4\big(\bm{\pi}^{(9)}\big)$. Similarly, we get $w_6(\bm{\pi}^*(\textbf{s}^*))=w_6\big(\bm{\pi}^{(9)}\big)$ and $w_8(\bm{\pi}^*(\textbf{s}^*))=w_8\big(\bm{\pi}^{(9)}\big)$.
  \item When $i=5$, node 5 is in the second column in Figure~\ref{fig:compare9} (a) and in the third column in Figure~\ref{fig:compare9} (b). Based on equation (\ref{equ_hi}) and Lemma~\ref{lemma_ai}, $a_5(\bm{\pi}^*(\textbf{s}^*))=d_I-1=a_5\big(\bm{\pi}^{(9)}\big)+1$. Because of Lemma~\ref{lemma_aibi}, $a_5(\bm{\pi}^*(\textbf{s}^*))+b_5(\bm{\pi}^*(\textbf{s}^*))=d_I+d_C+1-5=a_5\big(\bm{\pi}^{(9)}\big)+b_5\big(\bm{\pi}^{(9)}\big)$, then $b_5(\bm{\pi}^*(\textbf{s}^*))=b_5\big(\bm{\pi}^{(9)}\big)-1$. Hence, $w_5(\bm{\pi}^*(\textbf{s}^*))-w_5\big(\bm{\pi}^{(9)}\big)=\beta_I-\beta_C\geq 0$. Similarly, $w_3(\bm{\pi}^*(\textbf{s}^*))-w_3\big(\bm{\pi}^{(9)}\big)\geq 0$ and $w_7(\bm{\pi}^*(\textbf{s}^*))-w_7\big(\bm{\pi}^{(9)}\big)\geq 0$.
  \item When $i=9$, the 9th selected node is in cluster in Figure~\ref{fig:compare9} (a) and separate in Figure~\ref{fig:compare9} (b). Based on Lemma~\ref{lemma_ai} and Lemma~\ref{lemma_aibi}, $a_9(\bm{\pi}^*(\textbf{s}^*))$ and $w_9(\bm{\pi}^*(\textbf{s}^*))=a_9(\bm{\pi}^*(\textbf{s}^*))\beta_I+b_9(\bm{\pi}^*(\textbf{s}^*))\beta_C=(d_C-5)\beta_C$. With equation (\ref{equ_ci}), $w_9\big(\bm{\pi}^{(9)}\big)=(d_I+d_C+1-9)\beta_C=(d_C-5)\beta_C$. Hence, $w_9(\bm{\pi}^*(\textbf{s}^*))=w_9\big(\bm{\pi}^{(9)}\big)$.
 \end{itemize}
\end{example}
Example~\ref{exam_k9} shows that $w_i(\bm{\pi}^*(\textbf{s}^*))=w_i\big(\bm{\pi}^{(9)}\big)$ for $i=1,2,4,6,8,9$ and $w_i(\bm{\pi}^*(\textbf{s}^*))\geq w_i\big(\bm{\pi}^{(9)}\big)$ for $i=3,5,7$. Hence , $MC(\textbf{s}^*,\bm{\pi}^*(\textbf{s}^*))\geq MC_9^*$, and adding one separate node reduces the capacity. Based on equation (\ref{equ_w*}) and (\ref{equ_wn*}), the tradeoff bounds are plotted and compared in Figure~\ref{fig_tradeoffDiffk}~(b), showing that the tradeoff bounds after adding a separate node move left. As introduced in Subsection~\ref{subsec_PM} and proved in \cite{Dimakis2010}, the tradeoff bound is a lower bound for the region of feasible points $(\alpha, d_I, \beta_I, d_C, \beta_C)$ ($(\alpha, \beta_C)$ in Figure~\ref{fig:compare9}) to reliably store the original file of size $\mathcal{M}$. In this situation, adding one separate node reduces the feasible region of reliable storage points.

\begin{figure}[!t]
  \centering
  \subfloat[$k=7$]{\includegraphics[width=0.47\textwidth]{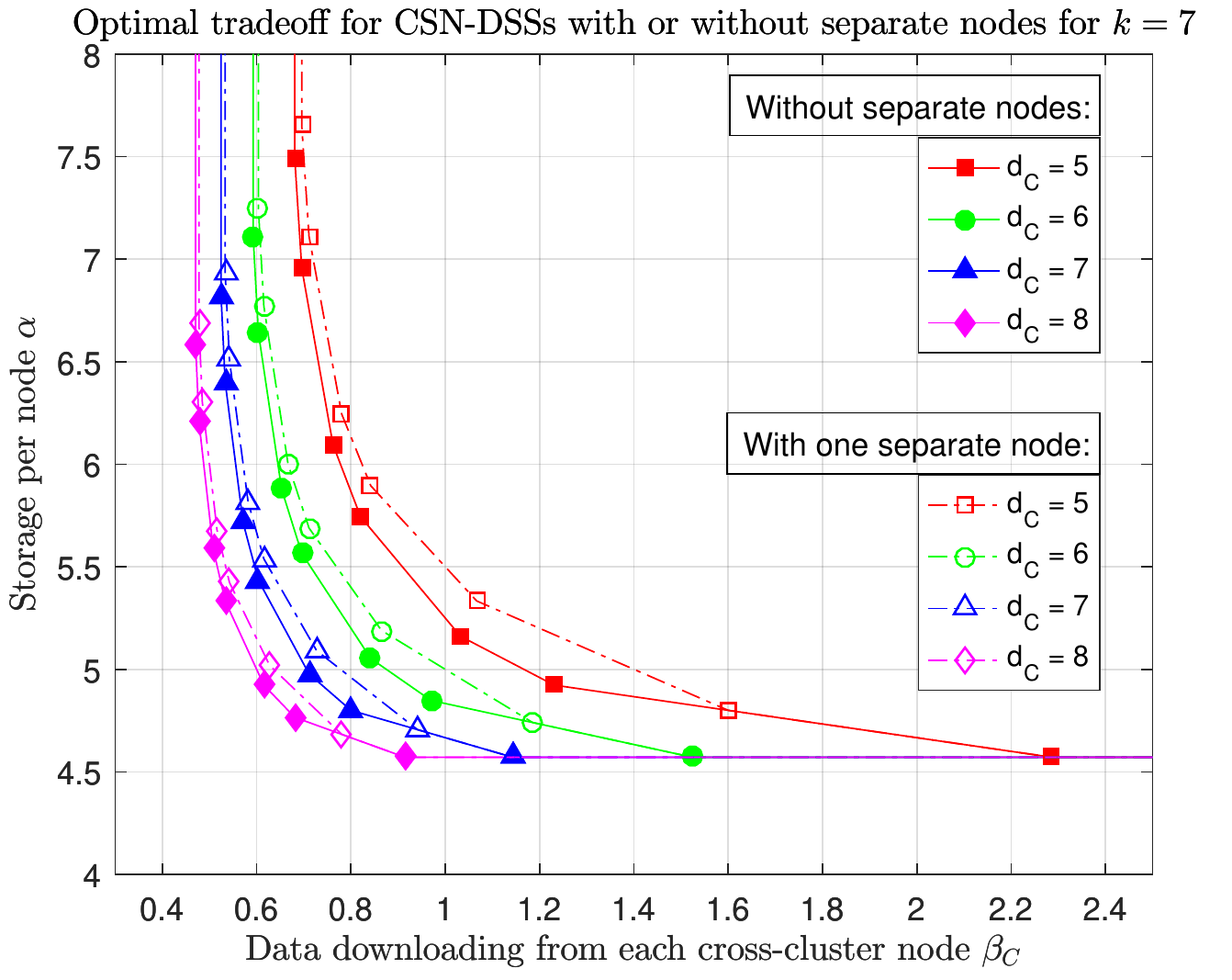}}
  \hspace{0.01\textwidth}
  \subfloat[$k=9$]{\includegraphics[width=0.47\textwidth]{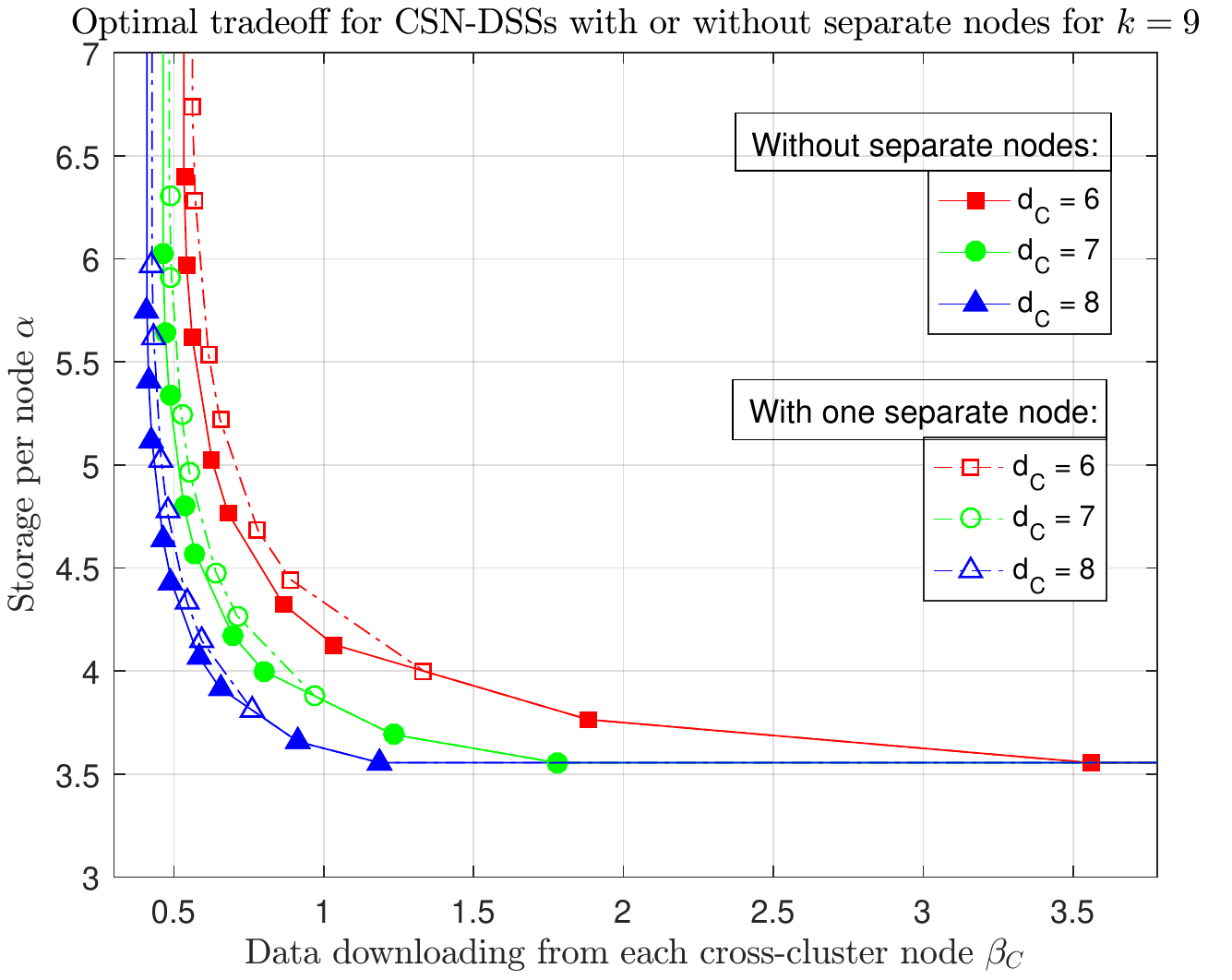}}
  \caption{Optimal tradeoff curves between node storage $\alpha$ and each cross-cluster bandwidth $\beta_C$ for the cluster DSS model with $(n=12,k=7\text{ or } 9,L=3,R=4,E=0)$ and the CSN-DSS model with $(n=13,k=7\text{ or }9,L=3,R=4,E=1)$. The bandwidth constraint is $\tau=\beta_I/\beta_C=2$ and $\mathcal{M}=32$. The tradeoff curves are for different numbers of cross-cluster helper nodes $d_C$ respectively. The solid and dotted lines are for the cluster DSS and CSN-DSS models respectively.}\label{fig_tradeoffDiffk}
\end{figure}

\begin{figure}[!htp]
  \centering
  \subfloat[$\bm{\pi}^*(\textbf{s}^*)={(1,2,1,2,1,2,1,2)}$]{\includegraphics[width=0.23\textwidth]{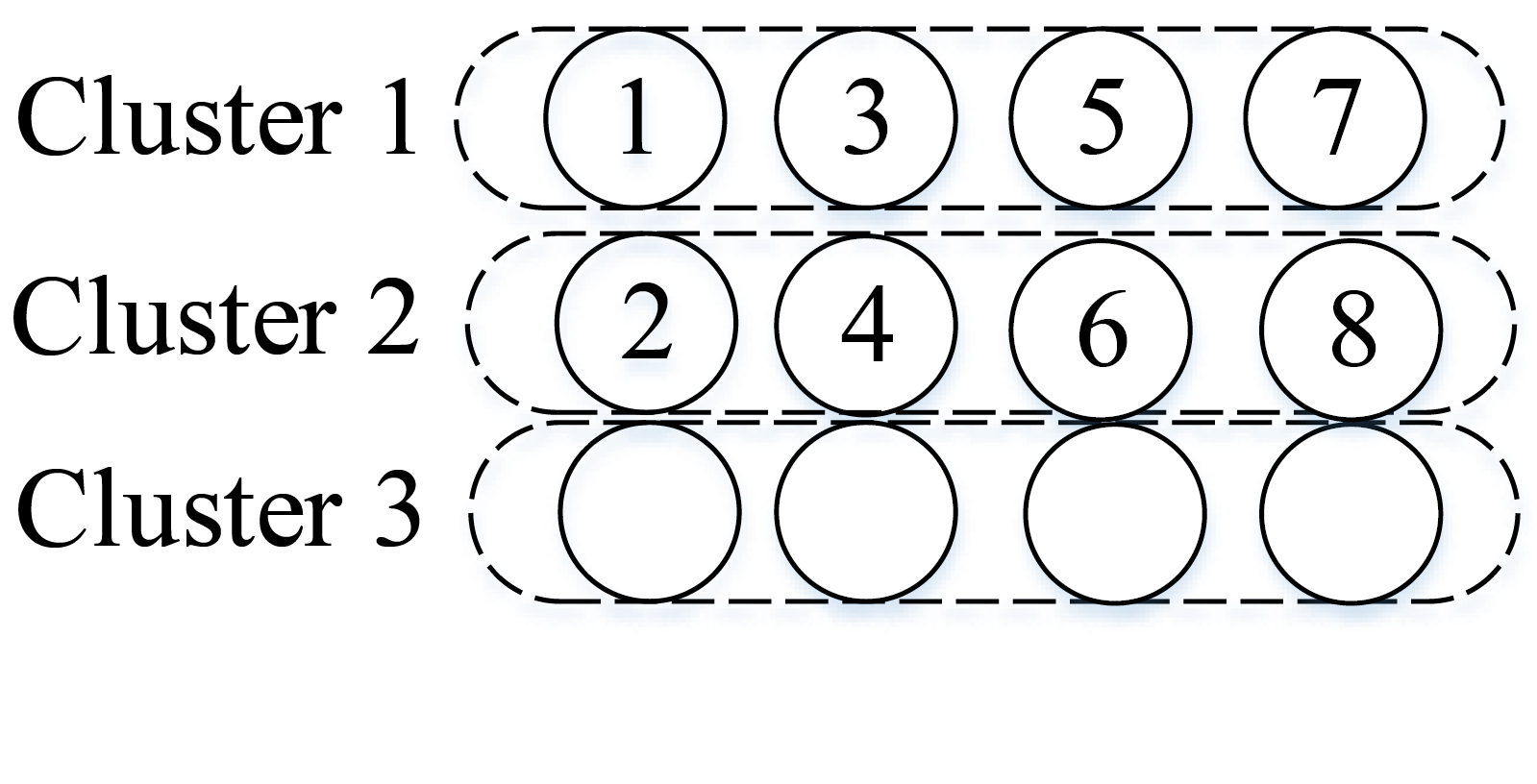}}
  \hspace{0.02\textwidth}
  \subfloat[$\bm{\pi}^{(8)}={(1,2,1,2,1,2,1,0)}$]{\includegraphics[width=0.23\textwidth]{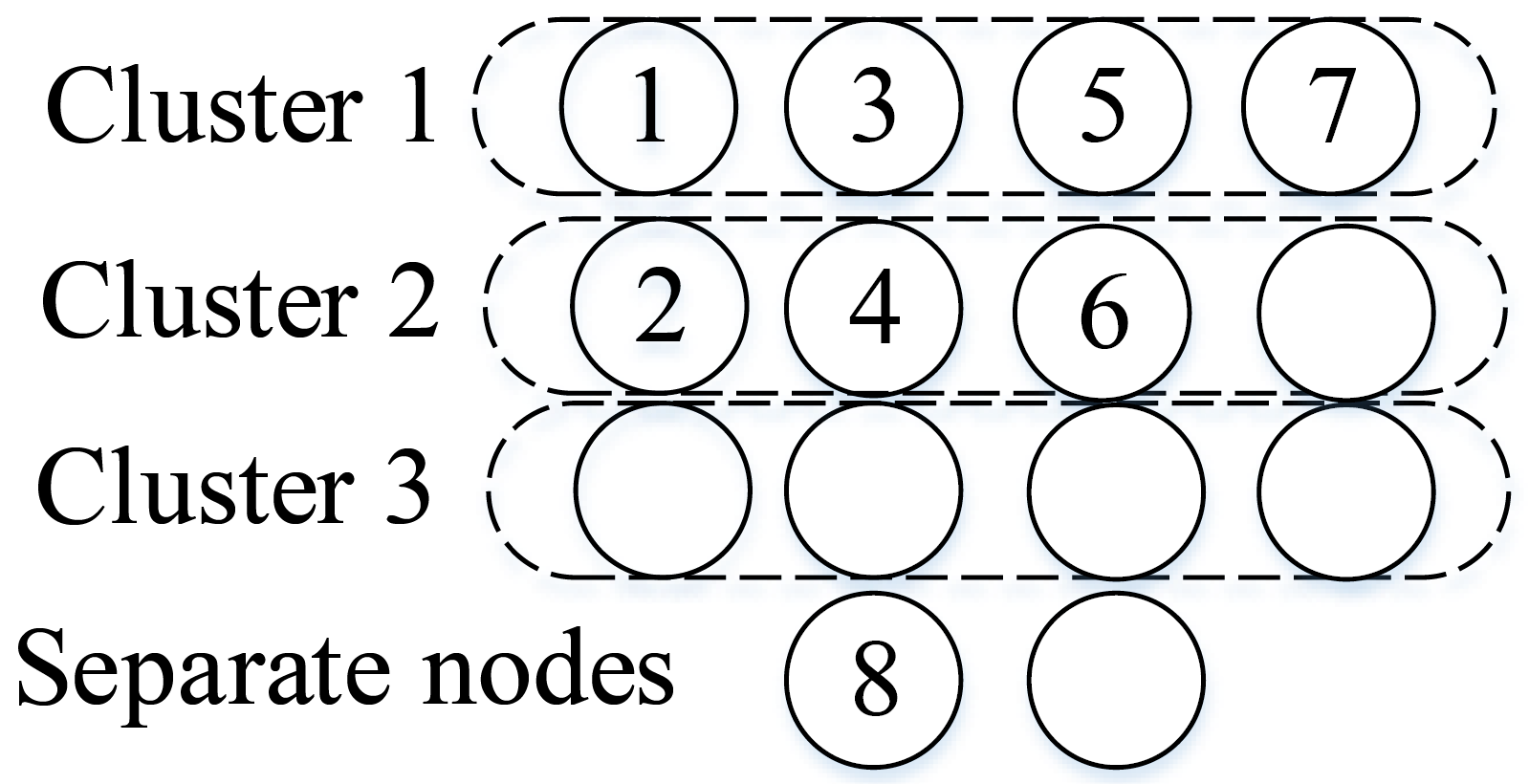}}
  \caption{In (a), the cluster order is $\bm{\pi}^*(\textbf{s}^*)$ for $\textbf{s}^*=(0,4,4,0)$, which achieves the capacity of the cluster DSS model. In (b), the cluster order and selected node distribution are $\bm{\pi}^{(8)}$ and $\textbf{s}^*=(1,4,4,0)$ respectively, which achieves the capacity of the CSN-DSS model.}\label{fig:compare8}
\end{figure}

\begin{example}\label{exam_k8}
Figure~\ref{fig:compare8} provides another example where adding a separate node will not change the system capacity. By comparing the selected nodes in $\bm{\pi}^*(\textbf{s}^*)={(1,2,1,2,1,2,1,2)}$ and $\bm{\pi}^{(8)}={(1,2,1,2,1,2,1,0)}$, it is easy to find that the locations of  the first 7 selected nodes in Figure~\ref{fig:compare8} (a) are the same to those in Figure~\ref{fig:compare8} (b). Hence, $w_i(\bm{\pi}^*(\textbf{s}^*))=w_i\big(\bm{\pi}^{(8)}\big)$ for $1\leq i\leq 7$ and we only need to compare $w_8(\bm{\pi}^*(\textbf{s}^*))$ and $w_8\big(\bm{\pi}^{(8)}\big)$. Based on Lemma~\ref{lemma_ai}, Lemma~\ref{lemma_aibi} and formula (\ref{equ_ci}),
\begin{align*}
&w_8(\bm{\pi}^*(\textbf{s}^*))\\
&=(d_I+1-h_{\bm{\pi}^*(\textbf{s}^*)}(8))\beta_I+(d_I+d_C+1-8-a_i(\bm{\pi}^*(\textbf{s}^*)))\beta_C\\
&=(R+d_C-7)\beta_C
\end{align*}
and
$$w_8\big(\bm{\pi}^{(8)}\big)=(d_I+d_C+1-8)\beta_C=(R+d_C-7)\beta_C,$$
where $d_I=R-1$ is given assumptions in this paper, as introduced before. Then, $w_8(\bm{\pi}^*(\textbf{s}^*))=w_8\big(\bm{\pi}^{(8)}\big)$. Hence, $MC(\textbf{s}^*,\bm{\pi}^*(\textbf{s}^*))=MC_8^*$ and adding one separate node will not change the capacity.
\end{example}

Example~\ref{exam_k9} and Example~\ref{exam_k8} illustrate that the node parameters decide whether adding one separate node will reduce the system capacity. In Theorem~\ref{theorem_compareS}, we investigate this problem theoretically.

\begin{theorem}\label{theorem_compareS}
  In a cluster DSS model with $(n,k,L,R,E=0)$ and the CSN-DSS model with $(n+1,k,L,R,E=1)$ after adding a new separate node, we assume the storage/bandwith parameters $(d_I=R-1,\ \beta_I,\ d_C,\ \beta_C)$ are the same for these two systems, achieving the reliable storage of file with size $\mathcal{M}$.
  \begin{itemize}
  \item If node parameters $R$ and $k$ satisfy
        \begin{align*}
            R \mid k,
        \end{align*}
  namely, $k$ is divisible by $R$, then the new added separate node will not change the capacity.
  \item If $R \nmid k$, the new added separate node will reduce the capacity of the original cluster DSS model.
  \end{itemize}
\end{theorem}
\begin{proof}
Let $\bm{\pi}^*(\textbf{s}^*)$ and $\bm{\pi}^{(k)}$ denote the cluster orders achieving the capacity of systems with and without the separate node, as shown in Proposition~\ref{prop_MC2} and Theorem~\ref{theorem_MCSN}. Through analysing cluster orders $\bm{\pi}^*(\textbf{s}^*)$ and $\bm{\pi}^{(k)}$, we compare the part incoming weights $[w_i(\bm{\pi}^*(\textbf{s}^*))]_{i=1}^k$ and $[w_i(\bm{\pi}^{(k)})]_{i=1}^k$ one by one and enumerate all possible cases. In the first part, we investigate the case $R\mid k$, where adding a separate node will not change the capacity corresponding to Example~\ref{exam_k8}. In the second part, we consider the case $R \nmid k$, where the system capacity is reduced (see Example~\ref{exam_k9}). See \ref{app_thm_compareS} for more details.
\end{proof}

Figure~\ref{fig_tradeoffDiffk} (a) and (b) present numerical comparisons between the tradeoff bounds of the cluster DSSs and CSN-DSSs after adding a separate node. It is shown that the tradeoff bounds move right after adding a separate node, implying that the capacities are reduced for different $d_C$ (the number of helper nodes), when $k=7,9$ and other parameters are the same $(\mathcal{M}=32,L=3,R=4,\tau=\beta_I/\beta_C=2)$.  Additionally, as $d_C$ increases, both of the tradeoff bounds with or without the separate node move left, meaning that the feasible region of the reliable storage points increase, which is consistent with the results in \cite{Dimakis2010} and \cite{Wangjz2018}.

In Theorem~\ref{theorem_compareS}, we prove that adding one separate node will reduce or keep the capacity of a cluster DSS, depending on the relationship of node parameters $R$ and $k$. This means that when $k$ is fixed, adding a separate node will not improve the system capacity. Additionally, this paper investigates the capacity of CSN-DSSs with one separate node. The capacity of CSN-DSSs with multiple separate nodes can be analysed similarly, but further theoretical proofs are needed.

\section{Conclusions}\label{sec_conclusion}

In this paper, we characterize the capacity of the CSN-DSS model with one separate node. The tradeoff bounds of CSN-DSSs are compared with cluster DSSs, which is instructive for construct flexible erasure codes adapting to various network conditions. A regenerating code construction strategy is proposed for the CSN-DSS model, achieving the minimum storage point in the tradeoff bound under specific parameters. The influence of adding a separate node is characterized theoretically. We prove that adding one separate node will reduce or keep the capacity of the cluster DSS model, depending on the system node parameters.

\appendices

\section{Proof of Theorem~\ref{theorem_MC3}}\label{app_thm_MC3}

We will reduce this proof to Proposition~\ref{prop_MC} and Proposition~\ref{prop_MC2}. As the separate selected node locations are the same for different $\bm{\pi}$ and $\textbf{s}$, we only need to consider the part of selected cluster nodes by analysing the influence of adding a separate selected node. It is convenient to represent the cluster order $\bm{\pi}$ with another cluster order $\overline{\bm{\pi}}$ without separate nodes, as formula (\ref{equ_newpi}) shows (see Figure~\ref{fig_SC} (b), (c)). The part incoming weights $w_i(\bm{\pi})$ $(1\leq i\leq k)$ are then expressed with $w_i(\overline{\bm{\pi}})$, and this theorem is proved by analysing $w_i(\overline{\bm{\pi}})$ with similar methods used in Proposition~\ref{prop_MC} and \ref{prop_MC2}.

\noindent \textbf{Cluster order assignment:} For any cluster order $\bm{\pi}$ with a separate selected node (the $j$-th one), we can always find a cluster order $\overline{\bm{\pi}}=(\overline{\pi}_1,\overline{\pi}_2,...,\overline{\pi}_k)$ without separate selected nodes satisfying that

\begin{equation}\label{equ_newpi}
  \pi_i=\begin{cases}
          \overline{\pi}_i, & \mbox{if } 1\leq i< j \\
          0, & \mbox{if } i=j \\
          \overline{\pi}_{i-1}, & \mbox{if } j<i\leq k
        \end{cases}.
\end{equation}
Note that the component $\overline{\pi}_k$ will not be used here, thus we only need to analyse the first $k-1$ components of $\overline{\bm{\pi}}$ actually. The corresponding selected node distributions are denoted by $\textbf{s}$ and $\overline{\textbf{s}}$ respectively.
\begin{figure*}[t]
  \centering
  \subfloat
  {\includegraphics[width=0.27\textwidth]{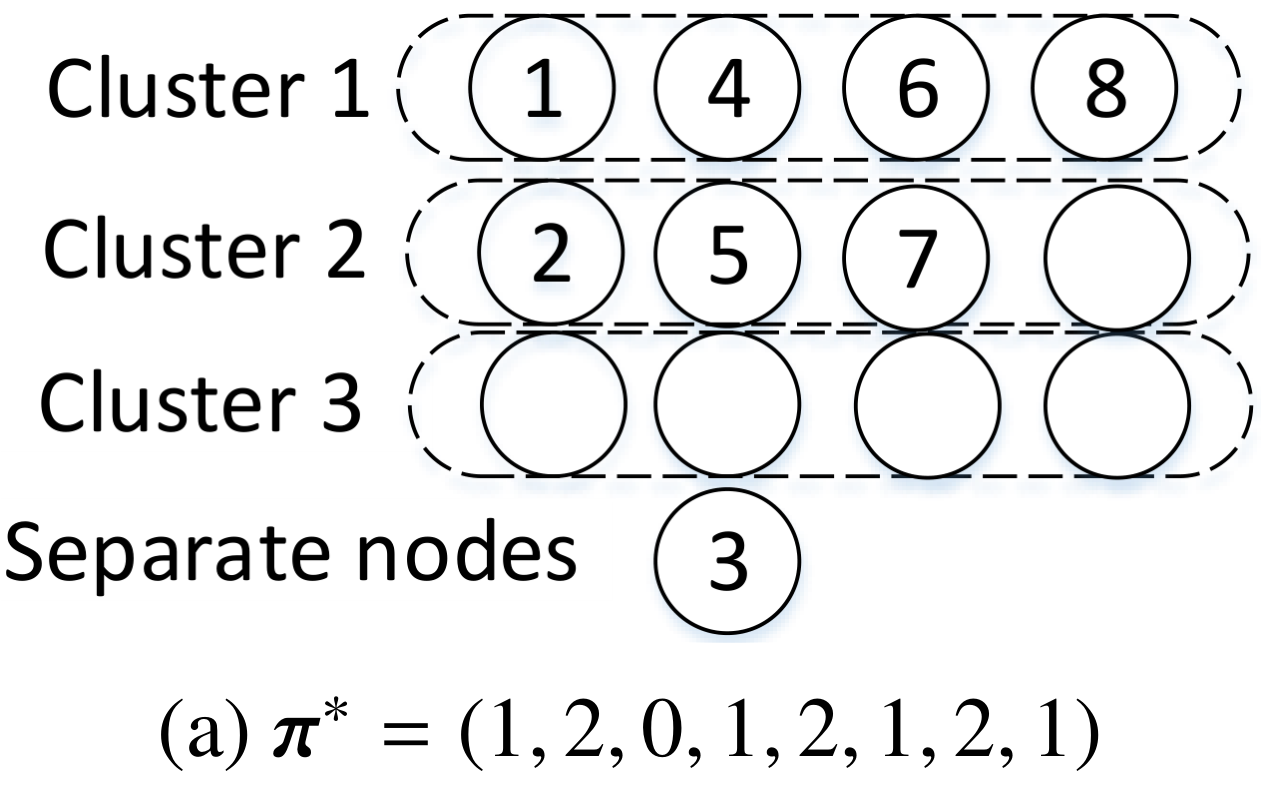}}
  \hspace{0.02\textwidth}
  \subfloat
  {\includegraphics[width=0.27\textwidth]{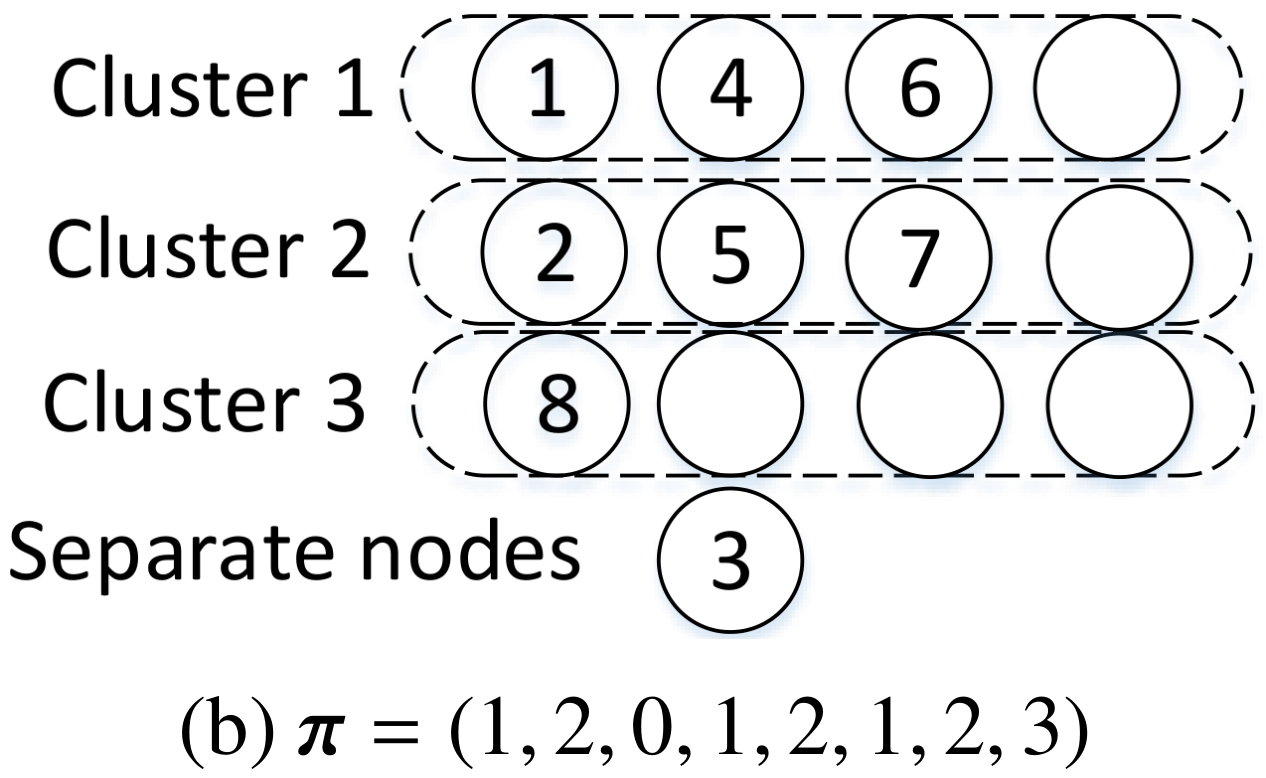}}
  \hspace{0.02\textwidth}
  \subfloat
  {\includegraphics[width=0.27\textwidth]{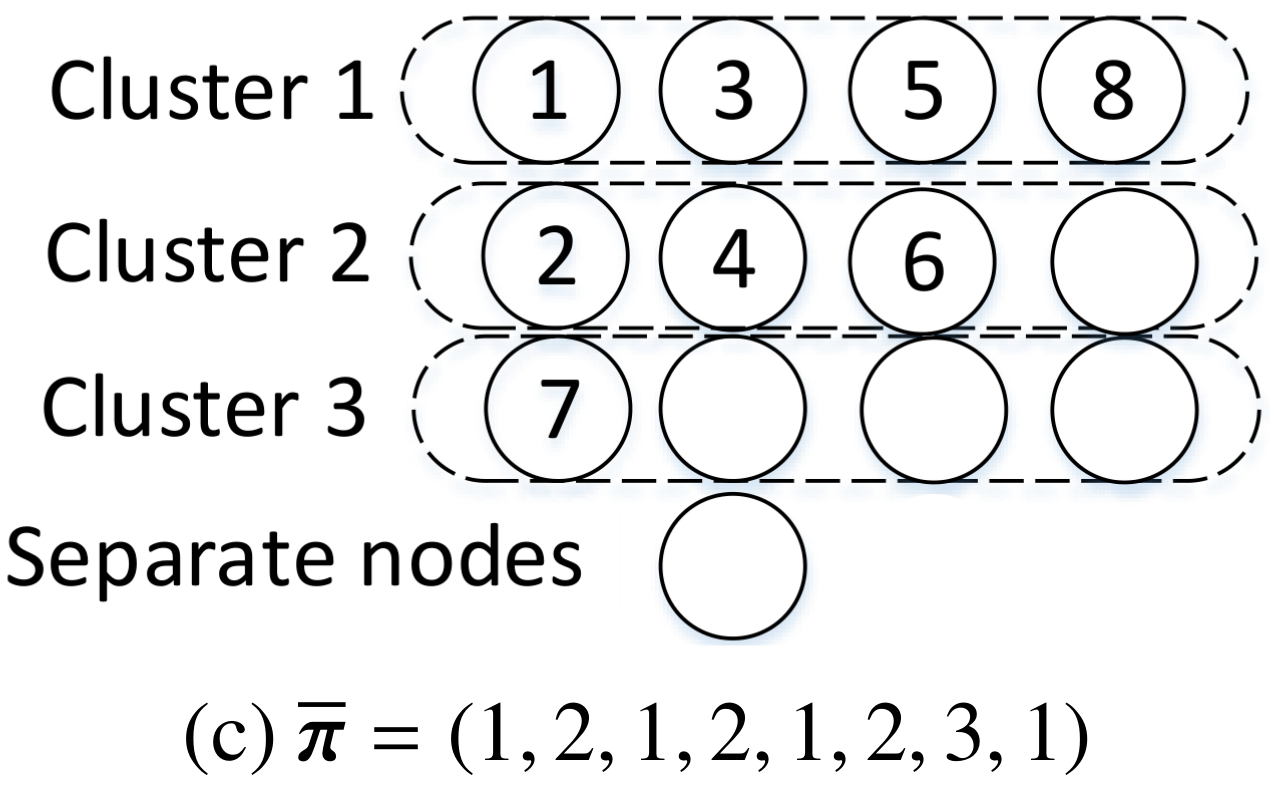}}
  \caption{Three cluster orders $\bm{\pi}^*$, $\bm{\pi}$ and $\overline{\bm{\pi}}$ for the CSN-DSS model.}\label{fig_SC}
\end{figure*}
For example, in Figure~\ref{fig_SC} (b), node 3 is separate in $\bm{\pi}=(1,2,0,1,2,1,2,3)$ which can be represented by the cluster order $\overline{\bm{\pi}}=(1,2,1,2,1,2,3,1)$ in Figure~\ref{fig_SC} (c) as $\pi_1=\overline{\pi}_1$, $\pi_2=\overline{\pi}_2$, $\pi_3=0$ and $\pi_i=\overline{\pi}_{i-1}(i=4,5,6,7,8)$. Figure~\ref{fig_SC} (a) shows the optimal cluster order $\bm{\pi^*}$ and selected node distribution$\textbf{s}^*$ generated by Algorithm~1 and 2, where the third separate selected node is fixed beforehand.

Based on (\ref{equ_wi}), the  part incoming weights for $\bm{\pi}$ are
\begin{small}\begin{equation}
  w_i(\bm{\pi})=\begin{cases}
                  w_i(\overline{\bm{\pi}}), & \mbox{if } 1 \leq i <j \\
                  (d_I+d_C+1-i)\beta_C, & \mbox{if } i=j \\
                  w_{i-1}(\overline{\bm{\pi}})-\beta_C, & \mbox{if } j+1\leq i\leq p\\
                  w_{i-1}(\overline{\bm{\pi}}), & \mbox{if } p+1\leq i\leq k
                \end{cases},
\end{equation}\end{small}\noindent
where $p$ is the integer that $b_{p}(\overline{\bm{\pi}})$ reduces to $0$. As the coefficient of $\beta_C$ won't be negative, the part incoming weight $w_i(\bm{\pi})=w_{i-1}(\overline{\bm{\pi}})=a_{i-1}(\overline{\bm{\pi}})\beta_I$ if $b_{i-1}(\overline{\bm{\pi}})=0$. If $j\geq p$, $w_i(\bm{\pi})=w_{i-1}(\overline{\bm{\pi}})$ for $j+1\leq i \leq k$. The value of $p$ varies according to $\bm{\pi}$ and $\overline{\bm{\pi}}$. Let $p_{max}$ denote the maximum value of $p$ for all $\bm{\pi}$ and $\overline{\bm{\pi}}$.

\noindent$\bullet$ For the case $j\geq p_{max}$,
\begin{small}\begin{align*}
&MC(\textbf{s}, \bm{\pi})=\sum_{i=1}^{k} \min \{w_i(\bm{\pi}), \alpha\}\\
&=\sum_{i=1}^{j-1} \min\{w_i(\overline{\bm{\pi}}), \alpha\}+\min\{c_j\beta_C, \alpha\}+\sum_{i=j+1}^k \min\{w_{i-1}(\overline{\bm{\pi}}), \alpha\}\\
&=\sum_{i=1}^{k-1} \min\{w_i(\overline{\bm{\pi}}), \alpha\}+\min\{c_j\beta_C, \alpha\}.
\end{align*}\end{small}\noindent
As the proof of Proposition~\ref{prop_MC} and Proposition~\ref{prop_MC2} does not depend on the value of $k$, it can be proved that
\begin{small}\begin{eqnarray*}
  \sum_{i=1}^k \min\{w_i(\bm{\pi}^*(\textbf{s}^*)), \alpha\}
  &=&\sum_{i=1}^{k-1} \min\big\{w_i\big(\overline{\bm{\pi}}^*(\textbf{s})\big), \alpha\big\}+\min\{c_j\beta_C, \alpha\} \\
  &\leq&\sum_{i=1}^{k-1} \min\{w_i(\overline{\textbf{s}},\overline{\bm{\pi}}), \alpha\}+\min\{c_j\beta_C, \alpha\}\\
  &=&\sum_{i=1}^k \min\{w_i(\textbf{s}, \bm{\pi}), \alpha\}
\end{eqnarray*}\end{small}

\noindent$\bullet$ When $j<p_{max}$, we will finish the proof in two parts corresponding to Proposition~\ref{prop_MC} and Proposition~\ref{prop_MC2}, respectively.

\noindent\textbf{Part 1:} For arbitrary given $\textbf{s}$ as Proposition~\ref{prop_MC} shows, we have
\begin{small}\begin{align*}
&MC(\textbf{s}, \bm{\pi})=\sum_{i=1}^{k} \min \{w_i(\bm{\pi}), \alpha\}\\
&=\sum_{i=1}^{j-1} \min\{w_i(\overline{\bm{\pi}}), \alpha\}+\min\{c_j\beta_C, \alpha\}\\
&+\sum_{i=j+1}^{p} \min\{w_{i-1}(\overline{\bm{\pi}})-\beta_C, \alpha\}+\sum_{i=p+1}^k \min\{w_{i-1}(\overline{\bm{\pi}}), \alpha\}\\
&=\sum_{i=1}^{j-1} \min\{w_i(\overline{\bm{\pi}}), \alpha\}+\min\{c_j\beta_C, \alpha\}\\
&+\sum_{i=j}^{p-1} \min\{w_{i}(\overline{\bm{\pi}})-\beta_C, \alpha\}+\sum_{i=p}^{k-1} \min\{w_{i}(\overline{\bm{\pi}}), \alpha\},
\end{align*}\end{small}\noindent
where $w_{i}(\overline{\bm{\pi}})=a_i(\overline{\bm{\pi}})\beta_I+b_i(\overline{\bm{\pi}})\beta_C$.
Let
\begin{eqnarray*}
w'_i(\overline{\bm{\pi}})=a'_i(\overline{\bm{\pi}})\beta_I  +b'_i(\overline{\bm{\pi}})\beta_C
=    \begin{cases}
       w_{i}(\overline{\bm{\pi}}), & \mbox{if } 1\leq i\leq j-1  \\
       w_{i}(\overline{\bm{\pi}})-\beta_C, & \mbox{if } j\leq i\leq p-1\\
       w_{i}(\overline{\bm{\pi}}), & \mbox{if } p\leq i\leq k-1
       \end{cases},
\end{eqnarray*}
then
  $a'_i(\overline{\bm{\pi}})=a_i(\overline{\bm{\pi}})$
and
$$b'_i(\overline{\bm{\pi}})=\begin{cases}
                            b_i(\overline{\bm{\pi}}), & \mbox{if } 1\leq i \leq j-1  \\
                            b_i(\overline{\bm{\pi}})-1, & \mbox{if } j\leq i\leq p-1 \\
                            0, & \mbox{if } p\leq i \leq k-1
                          \end{cases}.$$
Although $b'_i(\overline{\bm{\pi}})=b_i(\overline{\bm{\pi}})-1$ for $j\leq i \leq p-1$, the properties of the coefficient sequences will not change. Assume $p^*$ is the value satisfying that

$$b_{p^*}(\bm{\pi}^*)=0,\text{ and }b_{p^*-1}(\bm{\pi}^*)>0.$$
When $1\leq i\leq p^*$ and $j< p^*$,
\begin{equation}
  \phi_i(\overline{\bm{\pi}}^*)=\begin{cases}
                                  d-i+1, & \mbox{if } 1\leq i \leq j-1 \\
                                  d-i, & \mbox{if } j\leq i \leq p^*.
                                \end{cases}
\end{equation}
and
\begin{equation}
  \phi_i(\overline{\bm{\pi}})\geq\begin{cases}
                                  d-i+1, & \mbox{if } 1\leq i \leq j-1 \\
                                  d-i, & \mbox{if } j\leq i \leq p^*.
                                 \end{cases}.
\end{equation}
With similar methods of Proposition~\ref{prop_MC}, it can be proved that
\begin{small}\begin{equation}
  \sum_{i=1}^{k-1} \min\{w'_i(\overline{\bm{\pi}}^*), \alpha\}\leq \sum_{i=1}^{k-1} \min\{w'_i(\overline{\bm{\pi}}), \alpha\}.
\end{equation}\end{small}\noindent
Hence, $MC(\textbf{s}, \bm{\pi}^*)\leq MC(\textbf{s}, \bm{\pi})$.

\noindent\textbf{Part 2:} Next we will prove that $MC(\textbf{s}^*, \bm{\pi}^*(\textbf{s}^*)) \leq MC(\textbf{s}, \bm{\pi}^*(\textbf{s}))$ as proved in Proposition~\ref{prop_MC2}. With Algorithm~1, it's easy to verify that sequence $\big(w'_i\big(\overline{\bm{\pi}}^*(\textbf{s})\big),...,w'_{k-1}\big(\overline{\bm{\pi}}^*(\textbf{s})\big)\big)$ is non-increasing. Although the values of $b'_i\big(\overline{\bm{\pi}}^*(\textbf{s})\big)$ may change when $i>j$, the relation between $w'_i\big(\overline{\bm{\pi}}^*(\textbf{s})\big)$ and $w'_i\big(\overline{\bm{\pi}}^*(\textbf{s}^*)\big)$ won't be influenced. Similarly to the method used in Proposition~\ref{prop_MC2}, we can prove $w'_i\big(\overline{\bm{\pi}}^*(\textbf{s})\big)\geq w'_i\big(\overline{\bm{\pi}}^*(\textbf{s}^*)\big)$ for $1\leq i \leq k-1$.
Hence,
$$MC(\textbf{s}^*, \bm{\pi}^*(\textbf{s}^*)) \leq MC(\textbf{s}, \bm{\pi}^*(\textbf{s})) \leq MC(\textbf{s}, \bm{\pi}).$$

\section{Proof of Theorem~\ref{thm_MCj}}\label{app_lem_MCj}

\begin{figure}[t]
\centering
  \subfloat[$\bm{\pi}^{(4)}={(1,2,1,0,2,1,2,1)}$]{\includegraphics[width=0.23\textwidth]{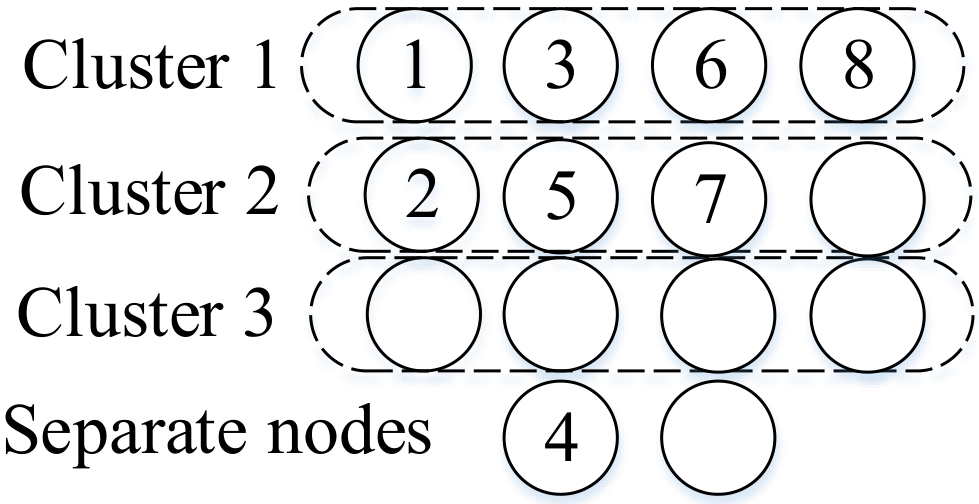}}
  \hspace{0.02\textwidth}
  \subfloat[$\bm{\pi}^{(5)}={(1,2,1,2,0,1,2,1)}$]{\includegraphics[width=0.23\textwidth]{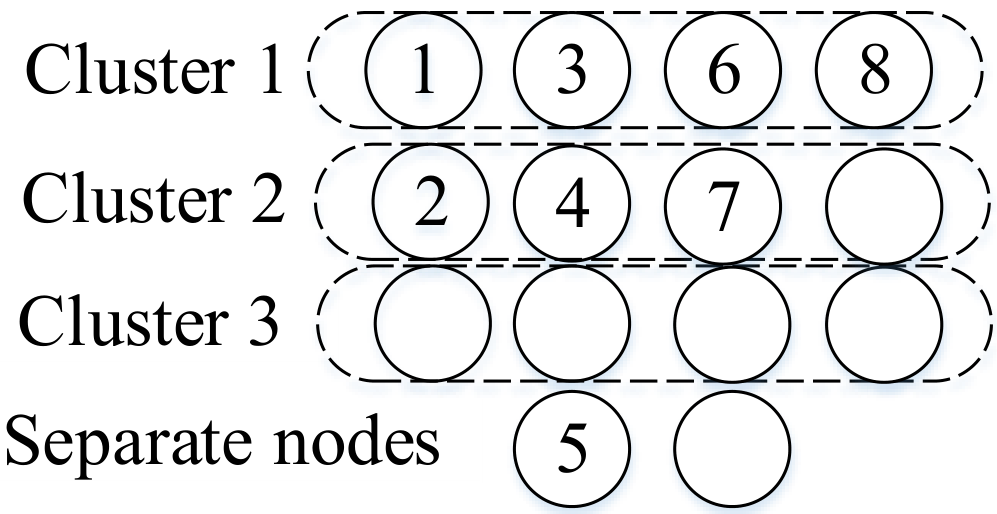}}
  \caption{The cluster orders $\bm{\pi}^{(4)}$ and $\bm{\pi}^{(5)}$ are corresponding to $\textbf{s}=(1,4,3,0)$, where the separate selected node locations are $4$ and $5$, respectively}\label{fig_lemmaOneSep}
\end{figure}

In this proof, we will compare $MC_j^*$ and $MC_{j+1}^*$ for $1\leq j\leq k-1$. As $MC_j^*=\sum_{i=1}^{k}\min\big\{w_i\big(\bm{\pi}^{(j)}\big), \alpha\big\}$, both $w_i\big(\bm{\pi}^{(j)}\big)$ and $\alpha$ need to be taken into consideration. Hence, this proof consists of two aspects. We first compare $\big[w_i\big(\bm{\pi}^{(j)}\big)\big]_{i=1}^k$ and $\big[w_i\big(\bm{\pi}^{(j+1)}\big)\big]_{i=1}^k$ one by one. Through analysing the cluster orders $\bm{\pi}^{(j)}$ and $\bm{\pi}^{(j+1)}$ (i.e., Figure~\ref{fig_lemmaOneSep}), we enumerate all the possible cases of $\big[w_i\big(\bm{\pi}^{(j)}\big)\big]_{i=1}^k$ and $\big[w_i\big(\bm{\pi}^{(j+1)}\big)\big]_{i=1}^k$. Second, we calculate and compare $MC_j^*$ and $MC_{j+1}^*$ concretely in \textbf{Case 1} and \textbf{Case 2} by considering $\alpha$. The details of this proof are as follows.

Through comparing $\bm{\pi}^{(j)}$ and $\bm{\pi}^{(j+1)}$, we can find
\begin{align}
w_i\big(\bm{\pi}^{(j)}\big)=w_i\big(\bm{\pi}^{(j+1)}\big),\label{equ_lc0}
\end{align}
for $i\in [k]\setminus \{j,j+1\}$\footnote{[k] represents the integer set $\{1,2,...k\}$.}. As Figure~\ref{fig_lemmaOneSep} shows, the first three and the last three selected nodes of $\bm{\pi}^{(4)}$ and $\bm{\pi}^{(5)}$ are at the same locations. Then we only need to compare
\begin{small}\begin{align*}
w_j\big(\bm{\pi}^{(j)}\big) &=(d_I+d_C+1-j)\beta_C,\\
w_{j+1}\big(\bm{\pi}^{(j)}\big) &= a_{j+1}\big(\bm{\pi}^{(j)}\big)\beta_I+b_{j+1}\big(\bm{\pi}^{(j)}\big)\beta_C
\end{align*}\end{small}\noindent
and
\begin{small}\begin{align*}
w_j\big(\bm{\pi}^{(j+1)}\big) &= a_j\big(\bm{\pi}^{(j+1)}\big)\beta_I+b_j\big(\bm{\pi}^{(j+1)}\big)\beta_C,\\
w_{j+1}\big(\bm{\pi}^{(j+1)}\big) &=(d_I+d_C-j)\beta_C.
\end{align*}\end{small}\noindent
Based on (\ref{equ_aibi*}) and $\beta_I \geq \beta_C$,
\begin{small}\begin{align}
w_{j+1}\big(\bm{\pi}^{(j)}\big)&\geq a_{j+1}\big(\bm{\pi}^{(j)}\big)\beta_C+b_{j+1}\big(\bm{\pi}^{(j)}\big)\beta_C \notag \\
&=(d_I+d_C-j)\beta_C \notag\\
&=w_{j+1}\big(\bm{\pi}^{(j+1)}\big), \label{equ_lc1}\\
w_{j}\big(\bm{\pi}^{(j+1)}\big)&\geq a_{j}\big(\bm{\pi}^{(j)}\big)\beta_C+b_{j}\big(\bm{\pi}^{(j)}\big)\beta_C \notag\\
&=(d_I+d_C+1-j)\beta_C \notag\\
&=w_{j}\big(\bm{\pi}^{(j)}\big).\label{equ_lc2}
\end{align}\end{small}\noindent
As proven in Lemma~\ref{lemma_ai}, $a_{j+1}\big(\bm{\pi}^{(j)}\big)=d_I+1-h_{\bm{\pi}^{(j)}}(j+1)$ and $a_{j}\big(\bm{\pi}^{(j+1)}\big)=d_I+1-h_{\bm{\pi}^{(j+1)}}(j)$. We can verify that the separate node will not influence the properties of $a_i$ and $b_i$. Hence,
\begin{align}
a_{j+1}\big(\bm{\pi}^{(j)}\big)=a_{j}\big(\bm{\pi}^{(j+1)}\big)
\end{align}
based on (\ref{equ_hi}), the definition of $h_{\bm{\pi}}(\cdot)$. In Figure~\ref{fig_lemmaOneSep}, $a_{5}\big(\bm{\pi}^{(4)}\big)=2=a_{4}\big(\bm{\pi}^{(5)}\big)$. Then
\begin{align}
b_{j}\big(\bm{\pi}^{(j+1)}\big)=b_{j+1}\big(\bm{\pi}^{(j)}\big)+1,
\end{align}
because of formula (\ref{equ_aibi*}) in Lemma~\ref{lemma_aibi}. Hence,
\begin{align}
w_j\big(\bm{\pi}^{(j+1)}\big)-w_{j+1}\big(\bm{\pi}^{(j)}\big)=\beta_C>0. \label{equ_lc3}
\end{align}
On the other hand, $w_j\big(\bm{\pi}^{(j)}\big)-w_{j+1}\big(\bm{\pi}^{(j+1)}\big)=\beta_C$. Thus
\begin{align}
w_j\big(\bm{\pi}^{(j+1)}\big)+w_{j+1}\big(\bm{\pi}^{(j+1)}\big)=w_j\big(\bm{\pi}^{(j)}\big)+w_{j+1}\big(\bm{\pi}^{(j)}\big)\label{equ_lc4}
\end{align}
With (\ref{equ_lc1}) (\ref{equ_lc2}) and (\ref{equ_lc3}), we can find that $w_j\big(\bm{\pi}^{(j+1)}\big)$ is the largest and $w_{j+1}\big(\bm{\pi}^{(j+1)}\big)$ is the smallest among $w_j\big(\bm{\pi}^{(j)}\big)$, $w_{j+1}\big(\bm{\pi}^{(j)}\big)$,$w_{j}\big(\bm{\pi}^{(j+1)}\big)$ and $w_{j+1}\big(\bm{\pi}^{(j+1)}\big)$. Hence, we only need to analyse the following two cases and consider $\alpha$.

\noindent\textbf{Case 1:} $w_j\big(\bm{\pi}^{(j+1)}\big) \geq w_{j}\big(\bm{\pi}^{(j)}\big)\geq w_{j+1}\big(\bm{\pi}^{(j)}\big) \geq w_{j+1}\big(\bm{\pi}^{(j+1)}\big)$.
\begin{enumerate}
  \item When $\alpha \geq w_j\big(\bm{\pi}^{(j+1)}\big)$, based on  (\ref{equ_lc4}),
  \begin{small}
  \begin{align*}
    &\min \big\{w_j\big(\bm{\pi}^{(j+1)}\big), \alpha \big\}+\min \big\{w_{j+1}\big(\bm{\pi}^{(j+1)}\big), \alpha \big\}\\
    &= w_j\big(\bm{\pi}^{(j+1)}\big) + w_{j+1}\big(\bm{\pi}^{(j+1)}\big)\\
    &= w_j\big(\bm{\pi}^{(j)}\big) + w_{j+1}\big(\bm{\pi}^{(j)}\big)\\
    &= \min \big\{w_j\big(\bm{\pi}^{(j)}\big), \alpha \big\} + \min \big\{w_{j+1}\big(\bm{\pi}^{(j)}\big), \alpha \big\}.
  \end{align*}
  \end{small}\noindent
  Hence, $MC_j^*= MC_{j+1}^*.$

  \item When $w_j\big(\bm{\pi}^{(j+1)}\big) \geq \alpha \geq w_{j}\big(\bm{\pi}^{(j)}\big)$,
  \begin{small}
  \begin{align*}
    &\min \big\{w_j\big(\bm{\pi}^{(j+1)}\big), \alpha \big\}+\min \big\{w_{j+1}\big(\bm{\pi}^{(j+1)}\big), \alpha \big\}\\
    &\leq w_j\big(\bm{\pi}^{(j+1)}\big) + w_{j+1}\big(\bm{\pi}^{(j+1)}\big)\\
    &\overset{(a)}{=} w_j\big(\bm{\pi}^{(j)}\big) + w_{j+1}\big(\bm{\pi}^{(j)}\big)\\
    &\overset{(b)}{=}\min \big\{w_j\big(\bm{\pi}^{(j)}\big), \alpha \big\} + \min \big\{w_{j+1}\big(\bm{\pi}^{(j)}\big), \alpha \big\},
  \end{align*}\end{small}\noindent
  where (a) results from (\ref{equ_lc4}) and (b) is based on $\alpha \geq w_{j}\big(\bm{\pi}^{(j)}\big) \geq w_{j+1}\big(\bm{\pi}^{(j)}\big)$. Hence,
  $$MC_j^*\geq MC_{j+1}^*.$$

  \item When $w_{j}\big(\bm{\pi}^{(j)}\big) \geq \alpha \geq w_{j+1}\big(\bm{\pi}^{(j)}\big)$,
  \begin{small}\begin{align*}
    &\min \big\{w_j\big(\bm{\pi}^{(j+1)}\big), \alpha \big\}+\min \big\{w_{j+1}\big(\bm{\pi}^{(j+1)}\big), \alpha \big\}\\
    &= \alpha + w_{j+1}\big(\bm{\pi}^{(j+1)}\big)\\
    &\overset{(a)}{\leq} \alpha + w_{j+1}\big(\bm{\pi}^{(j)}\big)\\
    &\overset{(b)}{=} \min \big\{w_j\big(\bm{\pi}^{(j)}\big), \alpha \big\} + \min \big\{w_{j+1}\big(\bm{\pi}^{(j)}\big), \alpha \big\},
  \end{align*}\end{small}\noindent
  where (a) and (b) are based on $ w_{j+1}\big(\bm{\pi}^{(j+1)}\big)\leq w_{j+1}\big(\bm{\pi}^{(j)}\big)\leq \alpha \leq w_{j}\big(\bm{\pi}^{(j)}\big)$. Hence,
  $$MC_j^*\geq MC_{j+1}^*.$$

  \item When $w_{j+1}\big(\bm{\pi}^{(j)}\big) \geq \alpha \geq w_{j+1}\big(\bm{\pi}^{(j+1)}\big)$,
  \begin{small}\begin{align*}
    &\min \big\{w_j\big(\bm{\pi}^{(j+1)}\big), \alpha \big\}+\min \big\{w_{j+1}\big(\bm{\pi}^{(j+1)}\big), \alpha \big\}\\
    &= \alpha + w_{j+1}\big(\bm{\pi}^{(j+1)}\big)\\
    &\overset{(a)}{\leq} \alpha + \alpha\\
    &\overset{(b)}{=} \min \big\{w_j\big(\bm{\pi}^{(j)}\big), \alpha \big\} + \min \big\{w_{j+1}\big(\bm{\pi}^{(j)}\big), \alpha \big\},
  \end{align*}\end{small}\noindent
  where (a) and (b) are based on $ w_{j+1}\big(\bm{\pi}^{(j+1)}\big)\leq \alpha \leq w_{j+1}\big(\bm{\pi}^{(j)}\big)\leq w_{j}\big(\bm{\pi}^{(j)}\big)$. Hence,
  $$MC_j^*\geq MC_{j+1}^*.$$
  \item When $w_{j+1}\big(\bm{\pi}^{(j+1)}\big) \geq \alpha$,
  \begin{small}\begin{align*}
    &\min \big\{w_j\big(\bm{\pi}^{(j+1)}\big), \alpha \big\}+\min \big\{w_{j+1}\big(\bm{\pi}^{(j+1)}\big), \alpha \big\}\\
    &= \alpha + \alpha\\
    &= \min \big\{w_j\big(\bm{\pi}^{(j)}\big), \alpha \big\} + \min \big\{w_{j+1}\big(\bm{\pi}^{(j)}\big), \alpha \big\},
  \end{align*}\end{small}\noindent
  and $MC_j^*= MC_{j+1}^*.$
\end{enumerate}

\noindent\textbf{Case 2:} $w_j\big(\bm{\pi}^{(j+1)}\big)\geq w_{j+1}\big(\bm{\pi}^{(j)}\big)\geq w_{j}\big(\bm{\pi}^{(j)}\big) \geq w_{j+1}\big(\bm{\pi}^{(j+1)}\big)$.

\noindent When $\alpha \geq w_j\big(\bm{\pi}^{(j+1)}\big)$, $w_j\big(\bm{\pi}^{(j+1)}\big) \geq \alpha \geq w_{j+1}\big(\bm{\pi}^{(j)}\big)$, $w_{j}\big(\bm{\pi}^{(j)}\big) \geq \alpha \geq w_{j+1}\big(\bm{\pi}^{(j+1)}\big)$ and $w_{j+1}\big(\bm{\pi}^{(j+1)}\big) \geq \alpha$, the situations are similar to Case 1: 1, 2, 4, 5, respectively, and can be analysed with the same methods. We only need to consider the situation that $w_{j+1}\big(\bm{\pi}^{(j)}\big) \geq \alpha \geq w_{j}\big(\bm{\pi}^{(j)}\big)\geq w_{j+1}\big(\bm{\pi}^{(j+1)}\big)$. Then
\begin{small}\begin{align*}
  &\min \big\{w_j\big(\bm{\pi}^{(j+1)}\big), \alpha \big\}+\min \big\{w_{j+1}\big(\bm{\pi}^{(j+1)}\big), \alpha \big\}\\
  &= \alpha + w_{j+1}\big(\bm{\pi}^{(j+1)}\big)\\
  &\leq \alpha + w_{j}\big(\bm{\pi}^{(j)}\big)\\
  &=\min \big\{w_{j+1}\big(\bm{\pi}^{(j)}\big), \alpha \big\} +\min \big\{w_j\big(\bm{\pi}^{(j)}\big), \alpha \big\}.
\end{align*}\end{small}\noindent
Hence, $MC_j^*\geq MC_{j+1}^*$ and we finish the proof.

\section{Proof of Theorem~\ref{theorem_compareS}}\label{app_thm_compareS}

  \begin{figure}[t]
  \centering
  \subfloat[$\bm{\pi}^*(\textbf{s}^*)={(1,2,1,2,1,2,1)}$]{\includegraphics[width=0.23\textwidth]{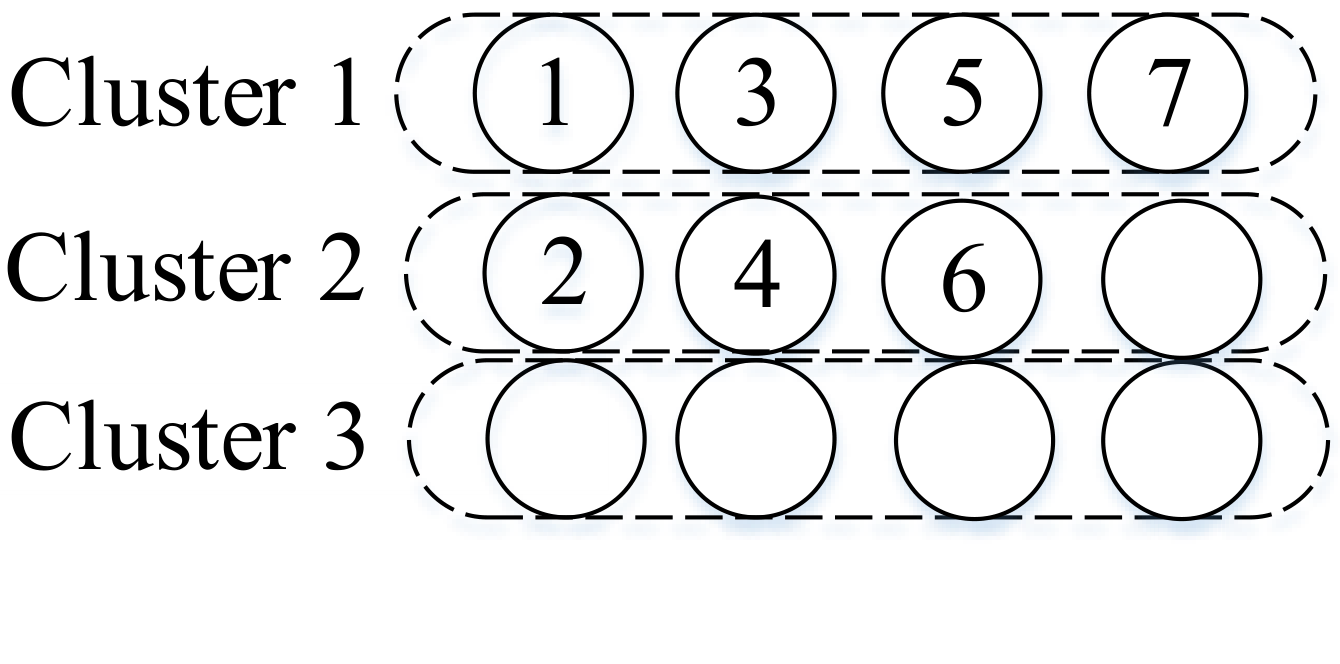}}
  \hspace{0.02\textwidth}
  \subfloat[$\bm{\pi}^{(7)}={(1,2,1,2,1,1,0)}$]{\includegraphics[width=0.23\textwidth]{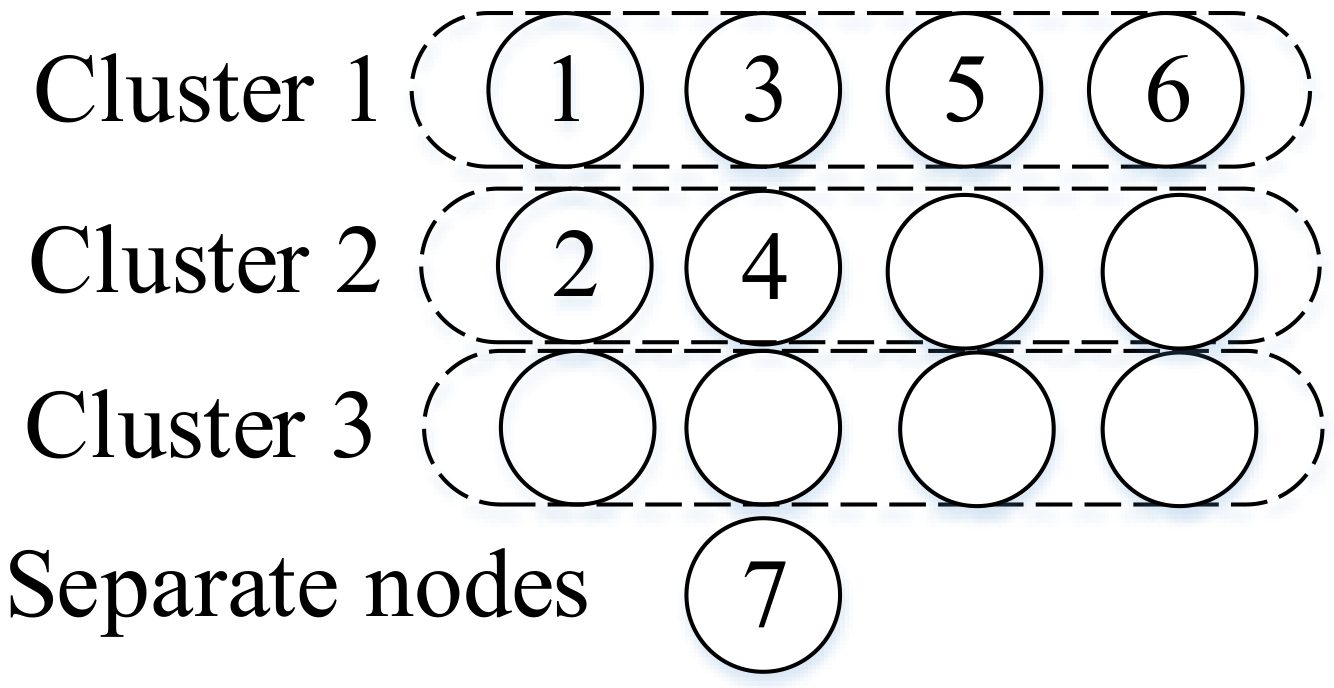}}
  \caption{In (a), the cluster order is $\bm{\pi}^*(\textbf{s}^*)$ for $\textbf{s}^*=(0,4,3,0)$, which achieves the capacity of the cluster DSS model. In (b), the cluster order and selected node distribution are $\bm{\pi}^{(7)}$ and $\textbf{s}^*=(1,4,2,0)$ respectively, achieving the capacity of the CSN-DSS model.}\label{fig:compare7}
\end{figure}

In this proof, we need to compare the capacities of the cluster DSS and the system after adding a separate node, which are represented by $\bm{\pi}^*(\textbf{s}^*)$ and $\bm{\pi}^{(k)}$ respectively, as shown in Proposition~\ref{prop_MC2} and Theorem~\ref{theorem_MCSN}. By analysing cluster orders $\bm{\pi}^*(\textbf{s}^*)$ and $\bm{\pi}^{(k)}$, we compare the part incoming weights $[w_i(\bm{\pi}^*(\textbf{s}^*))]_{i=1}^k$ and $[w_i(\bm{\pi}^{(k)})]_{i=1}^k$ one by one and enumerate all possible cases. We finish the proof in two parts. Part one investigates the case $R\mid k$, where adding a separate node will not change the capacity corresponding to Example~\ref{exam_k8}. In the second part, we consider the case $R \nmid k$, where the system capacity is reduced (see Example~\ref{exam_k9}).

\noindent\textbf{Part 1 ($R\mid k$):} As illustrated in Figure~\ref{fig:compare8}, the components of cluster order $\bm{\pi}^*(\textbf{s}^*)$ and $\bm{\pi}^{(k)}$ are the same except the $k$-th one. Hence, the part incoming weight (defined in Subsection~\ref{subsec_TermDef}) satisfies that
\begin{align*}
w_i(\bm{\pi}^*(\textbf{s}^*))=w_i\big(\bm{\pi}^{(k)}\big),
\end{align*}
for $1\leq i\leq k-1$. We only need to compare $w_k(\bm{\pi}^*(\textbf{s}^*))$ and $w_k\big(\bm{\pi}^{(k)}\big)$. When $R\mid k$, based on the horizontal selection algorithm, we can verify that $h_{\bm{\pi}^*(\textbf{s}^*)}(k)=R$. Because of  Lemma~\ref{lemma_ai} and Lemma~\ref{lemma_aibi}, $a_k(\bm{\pi}^*(\textbf{s}^*))=d_I+1-h_{\bm{\pi}^*(\textbf{s}^*)}(k)=0$
and
$b_k(\bm{\pi}^*(\textbf{s}^*))=d_I+d_c+1-k-a_k(\bm{\pi}^*(\textbf{s}^*))=R-k+d_C.$
On the other hand,
$w_k\big(\bm{\pi}^{(k)}\big)=(d_I+d_C+1-k)\beta_C=(R-k+d_C)\beta_C$
because of formula (\ref{equ_ci}). Hence, $w_k(\bm{\pi}^*(\textbf{s}^*))=0\beta_I+(R-k+d_C)\beta_C=w_k\big(\bm{\pi}^{(k)}\big)$. Therefore, adding one separate node will not change the system capacity.

\noindent\textbf{Part 2 ($R\nmid k$):} When $R\nmid k$, there are two cases: $k< R$ and $k\geq R$.
\begin{itemize}
  \item \textbf{Case 1} ($k<R$): In this case, all the selected nodes in $\bm{\pi}^*(\textbf{s}^*)$ are in Cluster 1. For $\bm{\pi}^{(k)}$, the first $k-1$ selected nodes are also in cluster 1. Hence, $w_i(\bm{\pi}^*(\textbf{s}^*))=w_i\big(\bm{\pi}^{(k)}\big)$ for $1\leq i \leq k-1$ and
      \begin{small}\begin{align*}
        w_k(\bm{\pi}^*(\textbf{s}^*)) &= a_k(\bm{\pi}^*(\textbf{s}^*))\beta_I + d_C\beta_C\\
        &\overset{(a)}{=}(d_I+d_C+1-k-d_C)\beta_I+d_C\beta_C\\
        &=(R-k)\beta_I+d_C\beta_C,\\
        w_k\big(\bm{\pi}^{(k)}\big)&= (d_I+d_C+1-k)\beta_C =(R-k+d_C)\beta_C,
      \end{align*}\end{small}\noindent
      where (a) is because of Lemma~\ref{lemma_aibi}. Hence,
      $w_k(\bm{\pi}^*(\textbf{s}^*))-w_k\big(\bm{\pi}^{(k)}\big)=(R-k)(\beta_I-\beta_C)\geq 0$
      and
      $\sum_{i=1}^k\min\{w_i(\bm{\pi}^*(\textbf{s}^*)),\alpha\}\geq\sum_{i=1}^k\min\{w_i\big(\bm{\pi}^{(k)}\big),\alpha\},$ indicating that adding one separate node reduces the system capacity.
  \item \textbf{Case 2} ($k\geq R$): As $R\nmid k$, let
  \begin{align}
    q&\triangleq \lfloor k/R\rfloor,\\
    r&\triangleq k \mod R.
  \end{align}
  Obviously, $1\leq q \leq L-1$ and $1\leq r \leq k-1$. All nodes in the first $q$ clusters and $r$ nodes in the $(q+1)$-th cluster are selected nodes based on the horizontal selection algorithm. Comparing each components of $\bm{\pi}^*(\textbf{s}^*)$ and $\bm{\pi}^{(k)}$, we can find that
  \begin{align}\label{equ_equalWi}
    w_i(\bm{\pi}^*(\textbf{s}^*))=w_i\big(\bm{\pi}^{(k)}\big)
  \end{align}
  for $1\leq i\leq (q+1)r-1$.

  In cluster order $\bm{\pi}^{(k)}$, the first $k-1$ selected nodes are cluster nodes, and there are $r-1$ selected nodes in the $(q-1)$-th cluster. As mentioned in Subsection~\ref{subsec_TermDef}, the nodes in cluster orders are numbered and grouped column by column. For cluster order $\bm{\pi}^{(k)}$, there are $q+1$ selected nodes in each of the first $r$ columns and $q$ selected nodes in each of the remaining $R-r$ columns. When a separate selected node is added, some of the numbers of selected cluster nodes will change. There are $q+1$ selected nodes in each of the first $r-1$ columns and $q$ selected nodes in the remaining $R-r+1$ columns in $\bm{\pi}^{(k)}$. As the selected cluster node are numbered from left to right column by column, the numbers of the first $(q+1)r-1$ selected nodes will not change, leading to equation (\ref{equ_equalWi}). An example is illustrated as follows.

  In Figure~\ref{fig:compare7}, the selected nodes are numbered by cluster orders $\bm{\pi}^*(\textbf{s}^*)$ and $\bm{\pi}^{(7)}$ in (a) and (b), respectively, where $R=4$,\ $k=7$, $q=\lfloor k/R \rfloor=1$ and $r=3$. In Figure~\ref{fig:compare7} (a), there are $q+1=2$ selected nodes in each of the first $r=3$ column and $1$ selected nodes in the last one columns. In Figure~\ref{fig:compare7} (b), there are $q+1=2$ selected nodes in each of the first $r-1=2$ columns and $q=1$ selected nodes in each of the remaining $R-r+1=2$ columns. We can verify that the locations of the first $(q+1)r-1=2\times 3-1=5$ selected nodes, namely node 1 to node 5, will not change. However, the locations of other nodes have changed. For example, node 6 is the third selected node in cluster 2 in $\bm{\pi}^*(\textbf{s}^*)$, but the 6th node in $\bm{\pi}^{(7)}$ is the 4th one in cluster 1, leading to the difference of $w_6(\bm{\pi}^*(\textbf{s}^*))$ and $w_6\big(\bm{\pi}^{(7)}\big)$.

  When $i=(q+1)r$, as shown in Figure~\ref{fig:compare7}, the $i$-th node in  $\bm{\pi}^*(\textbf{s}^*)$ is in a different column than that in $\bm{\pi}^{(k)}$ and
  $h_{\bm{\pi}^{(k)}}(i)=h_{\bm{\pi}^*(\textbf{s}^*)}(i)+1.$
  Hence,
  $a_i(\bm{\pi} ^*(\textbf{s}^*))=d_I+1-h_{\bm{\pi} ^*(\textbf{s}^*)}(i)=d_I+1-h_{\bm{\pi}^{(k)}}(i)+1=a_i\big(\bm{\pi}^{(k)}\big)+1.$
  Based on Lemma~\ref{lemma_aibi}, $a_i(\bm{\pi} ^*(\textbf{s}^*))+b_i(\bm{\pi} ^*(\textbf{s}^*))=a_i\big(\bm{\pi}^{(k)}\big)+b_i\big(\bm{\pi}^{(k)}\big)$, then $$b_i(\bm{\pi}^*(\textbf{s}^*))=b_i\big(\bm{\pi}^{(k)}\big)-1.$$
  Hence,
  $w_i(\bm{\pi} ^*(\textbf{s}^*))-w_i\big(\bm{\pi}^{(k)}\big)=\beta_I-\beta_C\geq 0.$

  For example, in Figure~\ref{fig:compare7}, $w_6(\bm{\pi} ^*(\textbf{s}^*))=(d_I-2)\beta_I+(d_C-3)\beta_C=\beta_I+(d_C-3)\beta_C$ and $w_6\big(\bm{\pi}^{(k)}\big)=(d_I-3)\beta_I+(d_C-2)\beta_C=(d_C-2)\beta_C$. Hence, $w_6(\bm{\pi} ^*(\textbf{s}^*))-w_6\big(\bm{\pi}^{(k)}\big)=\beta_I-\beta_C$.

  When $(q+1)r<i\leq k-1$, there are only two possible cases about the relationship of  node $i$ in $\bm{\pi} ^*(\textbf{s}^*)$ and $\bm{\pi}^{(k)}$.
  \begin{itemize}
    \item \textbf{Subcase 1:} The column number of node $i$ in $\bm{\pi} ^*(\textbf{s}^*)$ equals that in $\bm{\pi}^{(k)}$, namely, $h_{\bm{\pi} ^*(\textbf{s}^*)}(i)=h_{\bm{\pi}^{(k)}}(i)$. For example, in Figure~\ref{fig:compare9}, the 4th nodes in $\bm{\pi} ^*(\textbf{s}^*)$ and $\bm{\pi}^{(9)}$ are both in the second column. So $h_{\bm{\pi} ^*(\textbf{s}^*)}(4)=h_{\bm{\pi}^{(9)}}(4)=2$. Similarly, $h_{\bm{\pi} ^*(\textbf{s}^*)}(6)=h_{\bm{\pi}^{(9)}}(6)=3$ and $h_{\bm{\pi} ^*(\textbf{s}^*)}(8)=h_{\bm{\pi}^{(9)}}(8)=4$.

        Based on Lemma~\ref{lemma_ai} and Lemma~\ref{lemma_aibi}, we can get
        $$a_i(\bm{\pi} ^*(\textbf{s}^*))=a_i\big(\bm{\pi}^{(k)}\big)\text{ and }b_i(\bm{\pi} ^*(\textbf{s}^*))=b_i\big(\bm{\pi}^{(k)}\big).$$
        Therefore, $w_i(\bm{\pi} ^*(\textbf{s}^*))=w_i\big(\bm{\pi}^{(k)}\big)$.

    \item \textbf{Subcase 2:} The column number of node $i$ in $\bm{\pi} ^*(\textbf{s}^*)$ is smaller than that in $\bm{\pi}^{(k)}$ by 1, namely, $h_{\bm{\pi} ^*(\textbf{s}^*)}(i)=h_{\bm{\pi}^{(k)}}(i)-1$. In Figure~\ref{fig:compare9}, the 5th node in $\bm{\pi} ^*(\textbf{s}^*)$ is in the second column, while the 5th node in $\bm{\pi}^{(9)}$ is in the third column. Hence, $h_{\bm{\pi} ^*(\textbf{s}^*)}(5)=2=h_{\bm{\pi}^{(9)}}(5)-1$. Similarly, $h_{\bm{\pi} ^*(\textbf{s}^*)}(7)=3=h_{\bm{\pi}^{(9)}}(7)-1$.

        Based on Lemma~\ref{lemma_ai} and Lemma~\ref{lemma_aibi},
        $$a_i(\bm{\pi} ^*(\textbf{s}^*))=a_i\big(\bm{\pi}^{(k)}\big)+1\text{ and }b_i(\bm{\pi} ^*(\textbf{s}^*))=b_i\big(\bm{\pi}^{(k)}\big)-1,$$
        Hence, $w_i(\bm{\pi} ^*(\textbf{s}^*))-w_i\big(\bm{\pi}^{(k)}\big)=\beta_I-\beta_C\geq 0$
    \end{itemize}
    Combining the above 2 subcases, it is proved that $w_i(\bm{\pi} ^*(\textbf{s}^*))\geq w_i\big(\bm{\pi}^{(k)}\big)$ for $(q+1)r<i\leq k-1$.
    When $i=k$, based on Algorithm~1 and~2, the $k$-th selected node in $\bm{\pi} ^*(\textbf{s}^*)$ is the $R$-th one in its cluster, namely, $h_{\bm{\pi} ^*(\textbf{s}^*)}(k)=R$. Based on Lemma~\ref{lemma_ai} and Lemma~\ref{lemma_aibi},
    $a_k(\bm{\pi} ^*(\textbf{s}^*))=d_I+1-h_{\bm{\pi} ^*(\textbf{s}^*)}(k)=0$
    and
    $w_k(\bm{\pi} ^*(\textbf{s}^*))=a_k(\bm{\pi} ^*(\textbf{s}^*))\beta_I+(d_I+d_C+1-k-a_k(\bm{\pi} ^*(\textbf{s}^*)))\beta_C=(d_C+R-k)\beta_C.$
  As $w_k\big(\bm{\pi}^{(k)}\big)=(d_I+d_C+1-k)\beta_C=(d_C+R-k)\beta_C$, then $w_k(\bm{\pi} ^*(\textbf{s}^*))=w_k\big(\bm{\pi}^{(k)}\big)$.
\end{itemize}
Hence, when $R\nmid k$, $w_i(\bm{\pi} ^*(\textbf{s}^*))\geq w_i\big(\bm{\pi}^{(k)}\big)$ for $1\leq i\leq k$. Then $MC(\textbf{s}^*,\bm{\pi}^*(\textbf{s}^*))\geq MC\big(\bm{\pi}^{(k)}\big)$. We finish the proof.

\section*{Acknowledgment}
This work was partially supported by China Program of International S$\&$T Cooperation 2016YFE0100300, the National Natural Science Foundation of China under Grant 61571293 and SJTU-CUHK Joint Research Collaboration Fund 2018.



\bibliographystyle{plain}
\bibliography{DSbib}
%
%
%

\end{document}